\PassOptionsToPackage{dvipsnames}{xcolor}

\PassOptionsToPackage{unicode}{hyperref}

\documentclass[a4paper,UKenglish,cleveref, autoref, thm-restate]{article}
\usepackage{hyperref}

\usepackage{mathrsfs}  
\usepackage{csquotes}
\usepackage{amsmath}
\usepackage{amsfonts}
\usepackage{amssymb}
\usepackage{amsthm}
\usepackage{commath}
\usepackage{bm}
\usepackage{xfrac}
\usepackage{mathtools}
\usepackage{url}
\usepackage{bookmark}
\usepackage{float}
\usepackage{booktabs}
\usepackage{subcaption}
\usepackage{tikz}
\usetikzlibrary{calc,positioning,shapes,arrows}
\usepackage[dvipsnames]{xcolor}

\newcommand\MANONlim{ELim}

\newcommand{\composition}{\textit{composition}}
\newcommand{\linearLODE}{\textit{linear length ODE}}

\newcommand{\robODE}{\textit{robust continuous ODE}}
\newcommand{\contClasslim}{\bar{\mathbb{RCD}}}
\newcommand{\contClass}{\mathbb{RCD}}
\newcommand{\contClasslight}{\mathbb{RCD_{*}}}
\newcommand\encodageconfiguration{\gamma_{\textit{config}}}

\newcommand\manonclasslighttanh{\linearderivlength^\circ}

\newcommand\manonclasslighttanhlim{\bar{\manonclasslighttanh}}

\newcommand\send{\operatorname{send}}
\newcommand\smooth{\mathcal{C}}
\newcommand\sendtanh{\operatorname{\smooth-send}}
\newcommand\sendsymbol{\mapsto}

\newcommand\TTtanh{\operatorname{\smooth-if}}

\newcommand\relu{\operatorname{ReLU}}

\newcommand\sig{\operatorname{\mathfrak{s}}}
\newcommand\sigtanh{\operatorname{\mathfrak{\smooth-s}}}
\newcommand\relutanh{\operatorname{ReLU\mathfrak{-s}}}

\newcommand\distance[2]{d(#1,#2)}

\newcommand{\accept}{\textit{accept}}
\newcommand{\reject}{\textit{reject}}

\usetikzlibrary{cd}

\usepackage{versions}
\newcommand\dyadic{\mathbb{D}}

\newcommand\vectorl[1]{{\mathbf#1}}

\newcommand\vx{\vectorl{x}}
\newcommand\vy{\vectorl{y}}
\newcommand\vz{\vectorl{z}}

\newcommand\R{\mathbb{R}}
\newcommand\Q{\mathbb{Q}}
\newcommand\N{\mathbb{N}}
\newcommand\Z{\mathbb{Z}}

\newcommand{\dom}{\textit{Dom}}

\newcommand\projection[2]{\mathbf{\pi}_{#1}^{#2}}

\newcommand\plus{\mathbf{+}}
\newcommand\minus{\mathbf{-}}

\newcommand\gammaentier{\gamma_{\N}}
\newcommand\gammacompact{\gamma_{[0,1]}}
\newcommand\gammac{\Upsilon}
\newcommand\gammaentierk{\gamma_{\N}^{k}}
\newcommand\gammacompactk{\gamma_{[0,1]}^{k}}
\newcommand\gammacompactpk{\gamma_{[0,1]}^{'k}}

\newcommand\REACHPERTURBEDSPACE[2]{\reach^{#1}_{\{#2\}}}
\newcommand\mP{\mathcal{H}}

\newcommand{\tu}[1]{\mathbf{#1}}

\newcommand{\cp}[1]{\mathbf{#1}}
\newcommand{\Ptime}{\cp{PTIME}} 

\newcommand{\FPtime}{\cp{FPTIME}}

\newcommand{\Pspace}{\cp{PSPACE}}

\newcommand{\FPspace}{\cp{FPSPACE}}

\renewcommand\bar[1]{\overline{#1}}

\usepackage{color}

\newcommand{\dderiv}[2]{\frac{\delta #1}{\delta #2}}

\newcommand{\dderivl}[1]{\frac{\delta #1}{\delta \ell(x)}}

\newcommand\lengthnotation{\ell}
\newcommand{\length}[1]{\mathrm{\lengthnotation}(#1)}

\newcommand{\linearderivlength}{\mathbb{LDL}}

\usepackage{todonotes}

\newcommand\encodagemot{\gamma_{word}}
\newcommand\base{4}
\newcommand\symboleun{1}
\newcommand\symboledeux{3}

\newcommand\Image{\mathcal{I}}

\newcommand{\Next}{\textit{Next}}

\newcommand{\init}{\textit{init}}

\newcommand{\reach}{\textit{$\mathfrak{R}$}}

\definecolor{lightergray}{gray}{0.95} 
\usepackage{mdframed}
\newenvironment*{problems}[3]{\begin{mdframed}[linecolor= lightergray, linewidth = 2pt,leftmargin = 0.5pt , rightline = false, bottomline=false, topline = false]
		\textbf{#1}: \newline \underline{Input:} #2 \newline \underline{Output:} #3}
	{\end{mdframed}}

\newtheorem{theorem}{Theorem}

\theoremstyle{proposition}
\newtheorem{proposition}[theorem]{Proposition}

\theoremstyle{lemma}
\newtheorem{lemma}[theorem]{Lemma}

\theoremstyle{corollary}
\newtheorem{corollary}[theorem]{Corollary}

\theoremstyle{definition}
\newtheorem{definition}[theorem]{Definition}

\theoremstyle{remark}
\newtheorem{remark}[theorem]{NB}

\theoremstyle{observation}

\usepackage{wrapfig}
\usepackage{multicol}

\newcommand\THEPAPIERS{\cite{MFCS2019,MFCSJournal}}

\title{$\Pspace$-Completeness of the Reachability Relation of Robust Dynamical Systems}

	\usepackage{fancyhdr}

\usepackage{geometry}
\author{Manon BLANC, ITU i K{\o}benhavn, mbrb (at) itu.dk}
\date{ 
	version updated: September 2026}

\usetikzlibrary{decorations.pathreplacing,calligraphy}

\begin{document}
	\maketitle
	\begin{abstract}
		Dynamical systems are a widely used model of computation and an intuitive representation of many natural phenomena in many applied fields. The last decades have seen an impressive use of computers in studying and analysing them, with several visible theoretical breakthroughs.
	 The question of reachability for dynamical systems is central: given two points, is there a trajectory between them in the dynamical system? 
	 Reachability, in particular over the reals, is undecidable in the general case: it is only computably enumerable, even under reasonable hypotheses, e.g. for a computable dynamics in general. But it is known to become decidable, and even lies in $\Pspace$, under suitably quantified robustness assumptions. 
		 
	In this paper, we focus on $\Pspace$-completeness results over the reals, in the sense of computable analysis, regarding discrete-time and continuous-time \emph{robust} dynamical systems. 
	We prove that, under some additional and rather reasonable hypotheses, the reachability relation becomes even $\Pspace$-complete for both classical and second-order complexity \emph{second-order} complexity, for an adapted notion of reduction. As a side-effect, we also prove an improvement of an algebraic Ordinary Differential Equations-based characterisation of $\Pspace$, removing one of the base functions of the algebra.

	By doing so, we link results from dynamical systems and complexity over the reals, including the models of computable analysis and second-order complexity. 
	\end{abstract}

\underline{NB:} This is an updated version (update: September 2026). The first version was made available on my personal webpage in January 2026. The free version Grammarly was used as a grammatical checker in the introduction. The writing, the statements and their proofs are my own, and checked by me.

\underline{Acknowledgments:} I thank Olivier Bournez for rereading this paper and pointing out the typos and imprecisions.

\newpage

	\section{Introduction}

Dynamical systems are a model of computation widely used in many areas of fundamental and applied science. Their interactions with computations over the reals have various motivations (generalised shifts \cite{Moo91}, low-dimensional dynamical systems \cite{KCG94}, control theory \cite{BloTsi00}, analogue computing \cite{bournez2021survey}). Given a fixed  dynamical system, a common computational problem is reachability ($\textbf{Reach}$): deciding whether there is a path in the dynamical system reaching some given point $y$ starting from a given $x$. This decision problem is well-known to be computably enumerable (c.e.), but not decidable in the general case, even when the underlying dynamics is computable. 	It can be proven by embedding a Turing machine into a dynamical system; the machine stops if and only if the reachability relation holds.

However, with a proper notion of \emph{robustness,} as defined in Definition \ref{defRobust}, $\textbf{Reach}$ becomes decidable. Intuitively, this involved definition can be understood as insensitivity to an arbitrarily small perturbation. Instead of sending one point to one solution point, the perturbed dynamic function sends one point to an open set, in the topological sense. A system is robust if and only if each of those open sets intersects the (non-perturbed) solution point. For example, in game theory or in verification, even if \textbf{Reach} is not decidable in the general case, we can still certify that a subclass of dynamical systems can reach or not reach some subset of states.
The authors of \cite{asarin01perturbed} and \cite{gao2012delta} apply this concept mainly to decidability questions.
Considering dynamical systems stable under an infinitesimal perturbation is an old idea, rediscovered several times under different names.  Among them, we mention the concept of ``chain reachability'' (\cite{conley1978}, \cite{collins05}). In verification, the concept of infinitesimally perturbed dynamics has been explored to provide alternative semantics for certain models of computation; e.g. \cite{puri2000dynamical} for timed automata. 
The authors of \cite{CSL24} proved it is also possible to deal with complexity questions, and not only computability. With the proper quantification over the allowed perturbation, namely the radius of the open sets ($2^{-p(n)}$, $p$ a polynomial and $n\in \N$ related to some notion of size), they prove that $\textbf{Reach}$ is in $\Pspace$. This paper studies similar premises to those of \cite{CSL24}  where they mainly use the \emph{computable analysis}' model to define computability and complexity theory over the reals, since they handle general dynamical systems described by Ordinary Differential Equations (ODEs).

The main results of this paper are $\Pspace$-completeness \emph{over the reals} results. But many different models of computation exist for real numbers, and the proper definitions of complexity classes vary, depending on the model.
For discrete-time models of computations over the reals, the most famous approaches are computable analysis, based on Turing machines (TMs) (\cite{Tur36}, \cite{Wei00}) and algebraic models such as the Blum Shub Smale (BSS) model (\cite{BSS89}, \cite{BCSS98}). For the latter model, an extensive list of decision problems  $\exists\R$-complete has been recently given: see \cite{Schaefer24}. It includes problems coming from evolutionary game theory, such as properties of Nash equilibria, which are also usually described by dynamical systems. Other completeness results regarding the reachability of some equilibria higher in the complexity hierarchy exist: see, for example, \cite{etessami} and \cite{CSR21}.
Both models were tailored for very different applications, and we cannot unify existing models with the equivalent of a Church-Turing thesis. For example, computable functions in a computable analysis model are necessarily continuous (see e.g. \cite{Wei00}), while the BSS model intends to consider functions and problems over the polynomials that are not necessarily continuous. 
It is also explained by the fact that some models have not been introduced with the idea of corresponding to actual physical machines, but also to discuss abstract complexity (lower and upper bounds) for the associated problems. 
Using TMs, in a discrete-time framework, the authors of \cite{BlancBournezMFCS2023vantardise} prove it is possible to have algebraic characterisations of $\Pspace$, relying on discrete Ordinary Differential Equations (ODEs) (also called finite difference ODEs). 
For continuous-time models, one of the first machines was the Differential Analysers \cite{ulmann2020analog}, mathematically formalised by the General Purpose Analog Computer (GPAC) (\cite{Sha41}). 

There are also various recent approaches and models from deep learning, such as Neural ODEs \cite{chen2018neural,kidger2022neural}. 
Each model corresponds to a particular class of  ODEs.  E.g., the GPAC corresponds to polynomial ODEs \cite{GraCos03} and Neural ODEs are made by selecting the best solution among a parameterised class of ODEs (e.g. \cite{kidger2022neural}). With TMs,  the authors of \cite{ICALP24} prove it is possible to have algebraic characterisations of $\Pspace$, relying on continuous ODEs. In \cite{KawCook12}, the authors have $\Pspace$-completeness results by restricting themselves to compact domains.

\paragraph*{Our contributions}

In the general case, both in computable analysis and in second-order complexity over the reals, the problem of reachability is not decidable. We consider these frameworks so we can describe what polynomial space means more precisely in models of computation closer to the analogue world. 
\begin{itemize}
	 \item We prove, in Section \ref{secClassicRed}, the $\Pspace$-hardness of \textbf{Reach} in a robust dynamical system of a rational domain (included in $\Q$) (Definition \ref{defRobust}) (``$\textbf{Reach}$ Robust $\manonclasslighttanh_\textit{light}$'' in \textbf{Proposition \ref{propRRLDLhard}} and \textbf{Corollary \ref{RRLightComp}}). The usual notion of reduction is used here.
	 The $\Pspace$-membership is proved in \cite{CSL24, BlancBournezJoC}. Thus, it is $\Pspace$-complete: we prove the necessary complexity lower-bound. 
	 \item We show that, with robust oracle TMs (Definition \ref{defType2}), and second-order reductions over real functions (Definition \ref{defTopoRed}), we have the $\Pspace^2$-$\leq_m^2$-hardness (see Definition \ref{defPspaceCompKC}) of $\textbf{Reach}$, for both discrete-time (``\textbf{DRRL}'' in \textbf{Lemma \ref{DRRLHard}} and \textbf{Corollary \ref{DRRLComp}}) and continuous-time dynamics (``\textbf{CRRL}'' in \textbf{Theorem \ref{CRRLhard}} and \textbf{Corollary \ref{CRRLComp}}). With the membership results of \cite{CSL24, BlancBournezJoC}, it implies that both problems are $\Pspace^2$-$\leq_m^2$-complete.  Up to our knowledge, such questions 
	 within the frameworks of \cite{KawCook12}, \cite{BlancBournezMFCS2023vantardise} and \cite{ICALP24}, were open.
	 To sum-up the two previous items, we prove:
\hspace{-.3cm} \begin{center}
	 	\begin{tabular}{|c|c|c|c|}
	 		\multicolumn{1}{c|}{ } & \multicolumn{1}{c|}{Usual reduction} & \multicolumn{2}{|c|}{Second-order reduction} \\
	 		\hline
	 		& $\textbf{Reach}$ Robust $\manonclasslighttanh_\textit{light}$ & \textbf{DRRL} & \textbf{CRRL} \\
	 		\hline
	 		-Completeness & $\Pspace$ (\textcolor{Violet}{Cor. \ref{RRLightComp}}) & $\Pspace^2$ (\textcolor{Violet}{Cor. \ref{DRRLComp}}) & $\Pspace^2$ (\textcolor{Violet}{Cor. \ref{CRRLComp}}) \\
	 		\hline
	 	\end{tabular}
	 \end{center}
	 All these results are presented here on $\Q$ or $\R$ but are valid on $\Q^d$ and $\R^d$, for $d\in\N$: vectors can be projected on each of their coordinates.
	 \item As a side-effect, we prove that we can remove the cosine as a base function from the algebra of \cite{ICALP24}, which gives a characterisation of $\FPspace$ over the reals (\textbf{Theorem \ref{thNewRCD}}) with an algebra with a smaller basis. Thus, we are converging toward a minimal algebra for characterising polynomial space with ODEs.

\end{itemize}

\paragraph*{Organisation of the paper}

In Section \ref{secMathBack}, we introduce several tools and existing results that we rely on later. 
First, it includes some notion of dynamical systems and their computability properties, and the notion of robustness we consider here. Then we introduce \emph{second-order} complexity and previous results regarding complexity classes over the reals and hardness results. Then, 

in Section \ref{secClassicRed}, we prove that the reachability relation is $\Pspace$-complete, considering the framework of \cite{BlancBournezMFCS2023vantardise}, with the usual reduction between decision problems for this. 
 In Section \ref{secPSComplete}, we prove our main results about $\Pspace^2$-$\leq_m^2$-completeness of the reachability relation for robust discrete- and continuous-time dynamical systems, with a stronger notion of reduction, dealing with functions over the reals in second-order complexity (\cite{KawCook12}). We also prove that the base of the algebra of \cite{ICALP24} characterising $\FPspace$ can be minimised. 
 
\subparagraph{Notation.} 
For $a,b\in\R$, $a=_\epsilon b $ means $\|a-b\| < \epsilon$, for $\|.\|$ the sup-norm. For $A,B$ two sets, $f$ some function, $f\in B^A$ means $f: A \rightarrow B$. We define an algebra $[f_1, \cdots, f_n; o_1, \cdots, o_k]$ as the set of functions containing the functions $f_1, \cdots, f_n$, closed under the operators $o_1, \cdots, o_k$.

The following notations for \textbf{Reach} are used throughout the paper:

\begin{tabular}{|c||c|}
	\hline
	& Meaning \\
	\hline
	\hline
	$\reach^\mP$ & Reachability relation on $\mP$ \\
	\hline
	$\reach_{\epsilon}^\mP$ & Reachability relation on $\mP$ with an $\epsilon$ perturbation \\
	\hline
	$\reach_{\{q\}}^\mP $ & Reachability relation on $\mP$ with an $2^{-q(n)}$ perturbation, with $q$ a polynomial \\
	\hline
	$\reach_\omega^\mP$ & $\bigcap_{\epsilon>0} \reach_{\epsilon}^\mP $ \\
	\hline
\end{tabular}

\section{Mathematical Background and state of the art}\label{secMathBack}

We assume the reader is familiar with the basic concepts of \emph{classical} computability, including TMs, and complexity.
We also consider notions of computable analysis, see \cite{Ko91}, \cite{Wei00} for a more comprehensive study. The main notions are in Appendix \ref{secCA}.

\subsection{Discrete-Time Dynamical Systems}

As dynamical systems are widely used in fundamental and applied science, the relations between dynamical systems over the reals and computations have been the source of many works, such as \cite{Moo91}, \cite{AMP95}, \cite{bournez2021survey}. For now, we are interested in discrete-time dynamical systems:

\begin{definition}[Discrete-Time Dynamical System \cite{HSD03}]
	A discrete-time dynamical system $\mP$ is a tuple $(X, f)$, with $X$ a set, called the \emph{domain} of $\mP$, written $\dom(\mP)$, and $f: X \rightarrow X$ such that for all $\tu x \in X$, the dynamic starting from $x$ is defined by the sequence $\tu x_0 = \tu x$ and $\tu x_{n+1} = f(\tu x_n)$. In other words, $(\tu x_n)_{n\in\N} = (f^n(\tu x))_{n\in\N}$. 
\end{definition}

The main problems we consider here are the problems of \emph{point-to-point reachability} and \emph{point-to-set reachability} in a discrete-time dynamical system. The reachability problem determines whether there exists $n\in\N$ such that $y = f^n(\tu x)$ or $f^n(\tu x) \in Y $. 
We write $\reach^\mP(\tu x,\tu y)$ and $\reach^\mP(\tu x,Y)$ for the corresponding relations in the dynamical system $\mP$.

\begin{multicols}{2}
\begin{problems}
	{Reach($\tu x$,$\tu y$)}{A dynamical system $\mP = (X,f)$, $\tu x,\tu y \in X$}{$\reach^\mP(\tu x,\tu y)$}
\end{problems}

\begin{problems}
	{Reach($\tu x$,$Y$)}{A dynamical system $\mP = (X,f)$, $\tu x\in X$, $Y \subseteq X $}{$\reach^\mP(\tu x, Y)$}
\end{problems}
\end{multicols}

We notice that restricting $f$ to be a PAM, instead of just computable, does not improve the computability, where PAM stands for ``Piecewise Affine Map'',  a map of the form $f(\tu x) = \tu A_i \tu x + \tu b_i$, $\tu x \in P_i$, $i = 1 \cdots N$, where $A_i$ are rational $d\times d$-matrices, $b_i \in \Q^d$ and $P_i$ are convex rational polyhedral sets. 
There exist some \emph{classical} complexity results for time- and space-bounded dynamical systems.

Assuming $f$ preserves the rationals ($f(\Q)\subseteq \Q$), we have:

\begin{center}
	\begin{tabular}{|c|c|c|}
		\hline
		& $X\subseteq \Q$, $f$ computable & $X\subseteq \Q$, $f$ PAM \\
		\hline
		\textbf{Reach}($\tu x$,$y$) & c.e.-complete & c.e.-complete  \\
		\hline
		\textbf{Reach}($\tu x$,$Y$) & c.e.-complete & c.e.-complete  \\
		\hline
	\end{tabular}
\end{center}

To prove those hardness results, it is sufficient to reduce from the Halting problem. See e.g. \cite{Moo91, HKPV95, pouly} for the proofs.

Notice that the complexity of the problem crucially depends on the dynamical systems we consider:

\begin{center}
	\begin{tabular}{|c|c|c|}
		\hline
		& $X\subseteq \Q$, $f\in \FPtime$ & $X\subseteq \Q$, $f\in \FPspace$ \\
		\hline
		\textbf{Reach}($\tu x$,$y$) & $\Ptime$-complete & $\Pspace$-complete \\
		\hline
		\textbf{Reach}($\tu x$,$Y$) & $\Ptime$-complete & $\Pspace$-complete \\
		\hline
	\end{tabular}
\end{center}
\noindent where we assume time to be given in unary and $\textbf{Reach}$ to be time-bounded, meaning the computation time does not exceed some given $t\in \N$. To prove those results, it is sufficient to reduce from the time-bounded Halting Problem.

\subsection{Robustness for dynamical systems}\label{subRobust}

We give more details on the notion of robustness for a dynamical system of \cite{asarin01perturbed}, so the reachability problem becomes decidable.
The idea is to apply the paradigm of small perturbations.
Let $\mP= (X, \tu f)$ be discrete-time dynamical system. For any $\epsilon>0$, $\mP_{\epsilon}$ defines the $\epsilon$-perturbed system. Its trajectories are defined as sequences $\mathbf{x}_t$ satisfying $\distance{\mathbf{x}_{t+1}}{\tu f\left(\mathbf{x}_t\right)}<\epsilon$ for all $t$, and $d$ a distance. 
This non-deterministic system is considered as $\mP$ subjected to a noise of magnitude $\epsilon$. For convenience, we write $\tu y \in \tu f_{\epsilon}(\vx)$ as a synonym for $\distance{\tu f(\vx)}{\tu y} < \epsilon$. 
The reachability in the system $\mP_\epsilon$ is denoted by $\reach^\mP_\epsilon(\cdot, \cdot)$ and $B(\tu y, \epsilon)$ the (open) ball of radius $\epsilon$ centered on $\tu y$. 
Thus, we give the definition of the \emph{robust} reachability problem:

\begin{definition}[Robust \textbf{Reach} \cite{CSL24, BlancBournezJoC}]\label{defRobRe}
	With a discrete-time dynamical system $\mP$, we associate the relation $\reach^\mP_\epsilon(\tu x, \tu y)$, such that for two points $\tu x, \tu y \in \dom(\mP)$, $\epsilon = 2^{-p}$ for $p\in\N$, $\reach^\mP_\epsilon(\tu x, \tu y)=1$ if there exists a finite trajectory in $\mP_\epsilon$ from $\tu x$ to $B(\tu y, \epsilon)$.
\end{definition}
\begin{definition}[Discrete-Time Robust Dynamical System \cite{CSL24, BlancBournezJoC}]\label{defRobust}
	Define $\reach^\mP_\omega(\tu x, \tu y)$ iff $\forall \epsilon > 0$, $\reach^\mP_\epsilon(x, y)$. Without loss of generality, we assume that $\epsilon$ is of the form $2^{-p}$, $p\in\N$. We consider a discrete-time dynamical system $\mP$ \emph{robust} iff for all $\tu x, \tu y \in \dom(\mP)$, $ \reach^\mP_\omega(\tu x, \tu y) = \reach^\mP(\tu x,\tu y)$.
\end{definition}

All trajectories of a non-perturbed system $\mP$ are also trajectories of the $\epsilon$-perturbed system $\mP_{\epsilon}$. If $\epsilon_1<\epsilon_2$ then any trajectory of the $\epsilon_1$-perturbed system is also a trajectory of the $\epsilon_2$-perturbed system. By \cite{CSL24},	$\operatorname{\bf{Reach}}$ for robust discrete-time dynamical systems is co-c.e. It is c.e. by definition.
We deduce that $\textbf{Reach}$ for robust discrete-time dynamical systems is decidable.

\begin{definition}[Polynomially robust system]\label{defPolyRob}
	A discrete-time dynamical system $\mP$ is polynomially robust if there is a polynomial $q$, with $q(n) \to \infty$, such that $\reach^\mP_{2^{-q(n)}}(x,y) = \reach^\mP(x,y)$ for all $x, y \in \dom(\mP)$ of size at most $n$.
\end{definition}

\begin{theorem}[\cite{CSL24, BlancBournezJoC}] \label{thdirectionunpbis}  \label{ab:ab:ab}
	Let $\mP= (X,\tu f)$, with $\tu f\in\FPtime$ locally Lipschitz, $X$ a closed rational box ($X = \Pi_{i=1}^d[a_i, b_i] \subseteq\R$, $a_i, b_i\in \Q$) . Then, $\REACHPERTURBEDSPACE{\mP}{q} \subseteq \Q^{d}\times \Q^{d} \times \N \in \Pspace$.
\end{theorem}
A function $f\in\R^E$ is \emph{Lipschitz} if there exists $k\in\N$ such that for all $ x,y \in E$, $\left|f(x)-f(y)\right| \leq k\left|x-y\right| $,
and \emph{locally} Lipschitz if every $x \in E$ has a neighbourhood $U$ such that $f$ restricted on $U$ is Lipschitz.
The existing results regarding those systems are summed up here:

\noindent\begin{tabular}{|c|c|c|c|}
	\hline
	& $X\subseteq\R$, $\mP$ non-robust & $\mP$ robust & $\mP$ poly. robust, $\tu f\in\FPtime$ Lips. \\
	\hline
	\textbf{Reach}$(x,y)$ & Undecidable & Decidable \cite{asarin01perturbed} & $\in \Pspace$ \cite{CSL24, BlancBournezJoC} \\
	\hline
\end{tabular}

\subsection{Reductions in continuous settings and second-order}

To talk about computable analysis within our framework of \emph{second-order}, we could consider a variant of Turing machines in this context, namely Type-2 machines: machines with $k$ input and working tapes together with a type specification $\left(Y_1, \cdots, Y_k, Y_0\right)$ with $Y_i\in \{\Sigma^*, \Sigma^\omega \}$, giving the type (finite or infinite) for each input tape and output tape.
These machines contain some read-only input tapes, containing the inputs, which can be either finite or infinite words, a read-write working tape and one write-only output tape. The only difference with classical Turing machines is that we allow it to output an infinite word written forever on its (or one of its) write-only infinite output tape(s). 
However, Type-2 machines are not appropriate for general second-order \emph{complexity} theory. For example, taking $f: \{0,1\}^\omega \rightarrow \{0,1\}^\omega$, then for $u$ of length $n$, the associated machine has to go through at most $2^{n-1}$ possible solutions for $f(u)$, enforcing directly an exponential upper-bound. Further discussions about it are in \cite{KawCook12}.
Thus, in the rest of the paper, we consider the definition through \emph{oracle} machines:
\begin{definition}[Oracle TM \cite{Ko91}]\label{defType2}
A real function $f : \R \rightarrow \R$ is computable if there is a function-oracle TM $M$ such that for each $x \in \R$ and each $\varphi \in \text{CF}_x$ (the set of all Cauchy sequences binary converging to x), the function $\psi$ computed by $M$ with oracle $\varphi$ (i.e. $\psi(n) = M^\varphi (n)$) is in $\text{CF}_{f(x)}$. 
\end{definition}
Intuitively, an oracle Turing machine $M$ computes a function $f$ over $\R$, if, for any $x\in\dom(f)$, for any $\varphi \in \text{CF}_x$, $M$ with oracle $\varphi$, the following: given the output precision $2^{-n}$, given in the form of integer $n$, $M$ eventually outputs $d \in \mathbb{D}$ with $\|d - f(x)\| \leq 2^{-n}$.
Computability and complexity classes over the reals depend on the model of computation we consider, and also the representations of the reals: some complexity results hold for some representation (of the reals) but not for another (see \cite{Wei00}): 
a \emph{representation} of a set $S$ is a surjective function $\delta : {\textbf{Reg}} \rightarrow S$ (\cite{Wei00}, \cite{KawCook12}), with \textbf{Reg} the set of regular functions (Definition \ref{defSize}).
Here, we consider only equivalent representations (see Definition \ref{defRepEquiv} in Appendix \ref{secCA}).
We could represent real numbers with infinite sequences (in $\Sigma^\omega$). The problem with such representation is that there is no robust measure of size for space complexity. Following \cite{KawCook12}, we use $\textbf{Reg}$ (Definition \ref{defSize}) as a naming system for real numbers, sets and functions. Hence, the complexity is defined as the length of the word representing the real number, plus the size of the underlying oracle, also represented by a word.
The size of a regular function $\phi$, also denoted by $\left| \phi \right|$, is $|\phi|(|u|) = |\phi(u)|$ and $|s|$ denotes the length of the string $s$.

\begin{definition}[Regular Functions \cite{KawCook12}] \label{defSize}
	A function $f:\Sigma^{*}\rightarrow\Sigma^{*}$ is \emph{regular} iff, for all $u,v \in \dom(f)$ such that $|u|\leq|v|$, then $|f(u)|\leq|f(v)|$.
\end{definition}

 The usual notion of reduction only deals with sets of languages as inputs. Here, we have functions as inputs; thus, we do reductions for the complexity with \emph{pseudofunctors}.
   In this context, we say that complexity classes are defined in \emph{second-order}. More precisely, second-order \emph{polynomials} (\cite{KawCook12}) are defined inductively as follows: a positive integer and a variable are a second-order polynomial. If $P$, $Q$ are second-order polynomials, then so are $P + Q$, $P\cdot Q$, $P(L)$ where $L\in\N^\N$ is a type-1 (a total or partial map from $\N$ to $\N$) variable. A second-order polynomial $P$ is in $\left(\N^\N\right)^{\N^\N}$.

\begin{definition}[Second-order polynomials \cite{KawCook12}]
	A ((non-)deterministic) TM $M$ runs in \emph{second-order polynomial time} if there is a second-order polynomial $P$ such that, given any oracle $\phi\in\emph{\textbf{Reg}}$, any input $u\in\Sigma^{*}$, $M$ halts within $P(|\phi|)(|u|)$ steps. Define \emph{second-order polynomial space} analogously by counting the number of visited cells on all tapes. 
\end{definition}

\begin{definition}[Complexity Classes over $\R$]\label{defT2TCC}\label{defT2CCC}
	$\Pspace^2$ (resp. $\Ptime^2$) is the class of decision problems decidable in second-order polynomial space (resp. time) by an oracle TM. $\FPspace^2$ (resp. $\FPtime^2$) is the class of functions computable in second-order polynomial space (resp. time) by an oracle TM.
\end{definition}

We now define an adapted notion of reduction, with respect to oracle TMs (Definition \ref{defType2}).
Following \cite{KawCook12}, the reductions work on the set of regular functions $\textbf{Reg}$. 
For decision problems, we assume that the represented set $S$ 
is $\{0,1 \}^{\Sigma^{*}}$. We denote by $\textbf{Pred}$ the set of $\{0,1\}$-valued regular functions, as in \cite{KawCook12}:

\begin{definition}[Type-2 reductions \cite{KawCook12}]\label{defTopoRed}
	Let $\textbf{A}$ and $\textbf{B}$ be decision problems, or multi-functions from \emph{$\textbf{Reg}$} to \emph{$\textbf{Pred}$}. We write $\textbf{A} \leq_m^2 \textbf{B}$ if there exists $\mathfrak{s}, \mathfrak{t} \in \FPtime^2$ such that for any $\phi\in \dom(\textbf{A})$, we have $\mathfrak{s}(\phi)\in\dom(\textbf{B})$ and for $\theta\in  \textbf{B}\left(\mathfrak{s}(\phi)\right)$ the function $\theta \circ \mathfrak{t}(\phi)\in\textbf{A}\left(\phi\right)$.
\end{definition}

Intuitively, for $\textbf{A}$ and $\textbf{B}$ two problems, 
we have that $\dom(\textbf{A})$ is the set of functions, which are inputs of the problem $\textbf{A}$,
and $\textbf{A}\left(\phi\right)(x) = \left( \textbf{B}\left(\mathfrak{s}\left(\phi\right)\right)\circ \mathfrak{t}\left(\phi\right)\right)(x) $.
The function $\mathfrak{s}$ transforms the functional instances, and $\mathfrak{t}$ does the ``pre-treatement'' at each step of the computation.

\begin{definition}[$\Pspace^2$-$\leq_m^2$-completeness]\label{defPspaceCompKC}
	A problem $P$ is $\Pspace^2$-$\leq_m^2$-complete iff
		$P\in \Pspace^2$ and
		$P$ is $\Pspace^2$-$\leq_m^2$-hard: for all $P'\in \Pspace^2$, $P' \leq_m^2 P$.
\end{definition}

This construction seems somewhat artificial, but it is invariant under replacing representations with equivalent ones, and strong enough to capture various classical problems (\cite{KawCook12}), including:
\begin{problems}
	{Basic $\Pspace^2$}{ $<M, \bar \mu, \phi>$, with $M$ an oracle TM, $\bar\mu \in \textbf{Reg}$ such that $\bar\mu(u) = 0^{\mu(|u|)}$ where $\mu:\N\rightarrow\N$ a non-decreasing polynomial bounding the space of $M$, $\phi\in\textbf{Reg}$; a string $u$.}{Does $M$, with an oracle $\phi$, accept $u$?}
\end{problems}	

where, for $s$ a string, $s^{k}$ is the string with $s$ repeated $k$ times.

\begin{proposition}[\cite{KawCook12}]
	\textbf{Basic $\Pspace^2$} is  $\Pspace^2$-$\leq_m^2$-complete.
\end{proposition}


		\subsection{Previous results on complexity over the reals}
		
		The discrete derivative is denoted by $\delta$, defined as $\frac{\delta f}{\delta n}(n,\tu x) = f(n+1, \tu x) - f(n,\tu x)$, and the (continuous) derivative by $\partial$. A \emph{discrete ODE} is an ODE generated by a discrete derivative.
		We introduce algebraic characterisations of $\FPtime$ and $\FPspace$, given in \cite{BlancBournezMFCS2023vantardise}. They allow us to have a proper equational representation for the execution of a TM. 
		The main model of discrete ODEs to characterise $\FPtime$ are Linear Length ODEs. 	
		A length ODE is a particular case of a discrete ODE, where instead of having a derivative step of size 1, the derivative steps is of size the \emph{binary} length of $n$. The length function $\ell: \N \to \N$ maps some integer to the length of its binary representation. Thus, the derivative steps are of size the successive powers of 2. A more detailed definition is in Appendix \ref{lenODE}.		
		The size of the solution of a linear length ODE can never exceed a polynomial, as proven in \cite{BlancBournezMCU22vantardise}, ensuring the $\FPtime$-membership of solving such ODEs. 

Also, $\FPtime$ and $\FPspace$ are not stable under the general notion of limit. 
It means that the limit of a sequence of functions computable in polynomial time (resp. space) is not necessarily computable in polynomial time (resp. space). Thus, the notion of limit must be restricted by imposing that the constructed approximations converge \emph{exponentially} quickly to a solution:

\begin{theorem}[Characterisation of $\FPtime$, \cite{BlancBournezMFCS2023vantardise}]\label{thCaracTimeDis}\label{thmCharacPTIME}
	Let $\manonclasslighttanhlim = [\mathbf{0},\mathbf{1},\projection{i}{k}, \ell, \plus, \minus,\tanh,\frac{\cdot}{2},\frac{\cdot}{3};$ ${\composition, \textit{Linear Length ODE}, \MANONlim}],$ where $\projection{i}{k}$ outputs the $i^{th}$ coordinate of a vector of size $k$,
	$\tanh$ is the hyperbolic tangent function,
	$\frac{\cdot}{2}: \R \to \R$ divides by $2$ its argument (similarly for $\frac{\cdot}{3}$). The operator $\MANONlim$ is the \emph{effective limit}: given $\tilde{\tu f}:\N^{d+1} \to \R^{d'} \in \manonclasslighttanh =  \left(\manonclasslighttanhlim\backslash\{\MANONlim\}\right)$ such that
	for all $\tu m \in \N^{d}$, $n \in \N$,
	$\|\tilde{\tu f}(\tu m,~2^{n}) - \tu f(\tu m) \| \le 2^{-n},$ then 
	$\MANONlim(\tilde{\tu f})$ is the (uniquely defined) corresponding function  $\tu f: \N^{d} \to \R^{d'}$.
	 We have: 
	$\manonclasslighttanhlim \cap \R^{\N^{d} \times \R^{d'}}  = \FPtime \cap \R^{\N^{d} \times \R^{d}}.$
\end{theorem}
	
The effective limit of a sequence of polynomial-time (resp. -space) functions is computable in polynomial time (resp. space) (\cite{BlancBournezMCU22vantardise, BlancBournezMFCS2023vantardise}).
 All other basic functions are defined as usual, considered here as functions in ${\R}^{\R} $.

A characterisation of $\FPspace$ based on \emph{continuous} ODEs also exists:

\begin{definition}[Robust (continuous) ODE, \cite{ICALP24}]\label{defRobODE}
	
	A function $\tu f : \R \rightarrow \R$ is robustly ODE definable (from $\tu g$, $\tu{h}$, and dynamic $\tu u$) if:
	\begin{itemize}
		\item it corresponds to the
		solution of: $\tu f(0,\tu x) 
			= \tu g(\tu x)   \text{ and } \frac{\partial \tu f(t,\tu x)}{\partial t} 
			=   \tu u(\tu f(t,\tu x), \tu h(t,\tu x),
			t,\tu x);$
		\item there is some rational $\Delta >0$, 	and some polynomial $p$ such that  the previous schema is (polynomially) numerically stable on $[0, \Delta]$:
		for all integer $n$, considering $\epsilon(n)=2^{-p\left(n+\ell\left(\lceil\tu x\rceil\right)\right)}$ we can compute $\tu f(t,\tu x)$ by working with precision $\epsilon(n)$.
		Considering any solution of  
			$\tilde{\tu x} =_{\epsilon(n)} \tu x \quad \text{and} \quad \tilde{\tu h}(t,\tilde{\tu x}) =_{\epsilon(n)} {\tu h}(t,\tilde{\tu x})$,
		 and 
		$
		\tilde{\tu f}(0,\tilde{\tu x}) =_{\epsilon(n)}
		\tu g(\tu x)$  and 
		$
		\frac{\partial \tilde{\tu f}(t,\tilde{\tu x})}{\partial t} =_{\epsilon(n)}
		\tu u(\tilde{\tu f}(t,\tilde{\tu x}), \tilde{\tu h}(t,\tilde{\tu x}),
		t,\tilde{\tu x}) 
		$
		then
		$\tilde{\tu f}(t,\tilde{\tu x}) =_{\epsilon(n)} \tu f(t,\tu x)$ when  $0 \le t \le \Delta$;
		
		\item for $t \ge \Delta$,  we can compute $\tu f(t,\tu x)$ by computing some approximation $\widetilde{\tu f(t/2,\tu x)}$ of 
		$\tu f(t/2,\tu x)$ at precision  $\epsilon(n)$, i.e. we can compute $\Phi(\tu g(\tu x),t/2)$, by computing some approximation of  $\Phi(\widetilde{\tu f(t/2,\tu y)},t/2)$, working at precision  $\epsilon(n)$.
	\end{itemize}	
\end{definition}

\begin{theorem}[$\FPspace$ \cite{ICALP24}] \label{th:main:twopSpace} 
	$\contClasslim \cap \R^{\N^{d} \times \R^{d'}}= \FPspace \cap \R^{\N^{d} \times \R^{d}}$, where $\contClasslim = [0, 1,\pi_i^k, +,$ $ -, \times, \tanh, \cos, \pi, \frac{\cdot}{2}, \frac{\cdot}{3};$ $\composition,$ $ \robODE, \MANONlim].$
\end{theorem}

	\subsection{Properly encoding real numbers into a TM}\label{subsecEnco}

	A \emph{configuration} $C(q,l,r)$ of an oracle TM is such that $l,r \in \Sigma^{\omega}$ are 
	words over $\Sigma=\{0, \symboleun,\symboledeux\}$, $q \in Q$ denotes the internal state of $M$. The configuration of a bi-infinite tape TM $M$ is encoded by real numbers using radix \base{} encoding, but using only digits $\symboleun$,$\symboledeux$.
	Write: $\encodagemot: \Sigma^{\omega} \to \R$ for the function mapping a word $w=w_{0} w_{1} w_{2} \dots$ to the dyadic $\encodagemot(w)= \sum_{n \geq 0} w_n \base^{-(n+1)}$. For a configuration $C = (q, l, r)$, we write $\encodageconfiguration(C)$ for $(q, \encodagemot(l), \encodagemot(r))$. The proof of the following lemma is direct from the definition of $\encodagemot$:
	\begin{lemma}\label{lemGammaWordPtime}
		For $w = w_0 w_1 w_2 \cdots \in\Sigma^\omega$, for $n \in \N$, $\encodagemot(w_0w_1\cdots w_n)$ is computable in polynomial time in n.
	\end{lemma}

\begin{lemma}[\cite{BlancBournezMFCS2023vantardise}]\label{lemNext}
	We can construct some function $\bar {\Next}$ in $\manonclasslighttanh$ simulating one step of $M$: given a configuration $C$, writing $C'$ for the next one, for all integer $m$, $\| \bar{\Next}(2^{m},\bar C) - \bar C' \| \le 2^{-m}$, for $\bar C=\encodageconfiguration(C)=(q, \bar l,\bar r) \in \N \times \R^{2}$, where $\bar r= \encodagemot(r)$ and $\bar l=\encodagemot(l)$. 
\end{lemma}

We assume that $M$ halts on every input within $T_{\max} = |Q|\cdot S\cdot 3^{2S}$ steps (a machine working in space $S$ can be clocked within space $O(S)$), and that it erases its tape before halting, so that it has exactly two halting configurations, $C_{\text{accept}} = (q_{\text{accept}}, 0^\omega, 0^\omega)$ and $C_{\text{reject}} = (q_{\text{reject}}, 0^\omega, 0^\omega)$, encoded by $(q_{\text{accept}},0,0)$ and $(q_{\text{reject}},0,0)$. On a halting state $q$, we let $\Next(q,l,r) = (q, l/4, r/4)$, the branch being selected by $\TTtanh$ (the continuous approximation of an \textit{if}) on the state component. The function $\overline{\Next}$ remains in $\manonclasslighttanh_\text{light}$ and Lipschitz, and Lemma \ref{lemNext} still holds.
A proof of Lemma \ref{lemNext} is in Appendix \ref{AppNext}. Note that the construction of $\bar{\Next}$ involves discontinuous functions such as the floor function ($x\in\R \mapsto \lfloor x \rfloor = n \in \N $ such that $n \leq x < n+1$). But, since we are in the framework of computable analysis, we must deal only with continuous functions. The main motivation behind the (cryptic) use of only symbols $\symboleun$ and $\symboledeux$ in the base 4 is that it guarantees an expression like $\lfloor 4 \bar r \rfloor$, either equals $0$ (the specific case where there remain only blanks in $r$) or $\base \bar r$ is in the interval $[\symboleun,2]$ or in $[\symboledeux,4]$. 
From Theorem \ref{thmCharacPTIME} and Lemma \ref{lemNext}, we can deduce:

\begin{corollary}\label{coroTransPtime}
	An oracle TM's transition function is computable in second-order polynomial time.
\end{corollary}

	\section{$\Pspace$-completeness for the robust $\textbf{Reach}$ with ``classical'' TMs}\label{secClassicRed}

$\textbf{Reach}$, in the case where the underlying dynamical system is robust and the dynamic function is Lipschitz polynomial-time computable, is in $\Pspace$. 
Now, it is interesting to prove that this decision problem is even $\Pspace$-complete, hence has a lower-bound result. We consider here the framework of \cite{BlancBournezMFCS2023vantardise}, which gives a simple description of polynomial space with discrete ODEs. Here, we use the usual reduction between decision problems for this: $\textbf{A}\leq \textbf{B}$ if there is a computable in polynomial time function $f$ such that on every $w$, $w \in \dom(\textbf{A})$ is equivalent to $f(w) \in \textbf{B}$, as defined in \cite{Sip97}. It is sufficient, since rationals can be represented by tuples of integers. 
The algebra $\manonclasslighttanh = \manonclasslighttanhlim\backslash\{ \MANONlim\}$ allows us to talk about $\Ptime$ and $\Pspace$ with \emph{smooth}, even analytic, functions. We consider dynamic functions preserving the rationals on a rational domain. 
Let us define $\mathbb{LDL}^\circ_\textit{light} = \manonclasslighttanhlim\backslash\{  \ell; \linearLODE, \MANONlim\}= \left[\mathbf{0},\mathbf{1},\projection{i}{k},   \plus, \minus,\tanh,\frac{\cdot}{2},\frac{\cdot}{3}; \composition\right] $.
The algebra $\mathbb{LDL}^\circ_\textit{light}$ allows us to construct smooth approximations of the sigmoid, the transition and the \textit{if} functions ($\sigtanh, \sendtanh, \TTtanh$) only by composition (see \cite{BlancBournezMFCS2023vantardise} and Appendix \ref{AppNext}).

We first prove the $\Pspace$-completeness of the following problem:

\begin{problems}
	{Reach Robust $\manonclasslighttanh_\textit{light}$}{A discrete-time polynomially robust dynamical system $\mP = (X,g)$, $X$ a closed rational box, $g\in \mathbb{LDL}^\circ_\text{light}$, $Y\subseteq X$ a closed rational box, $\tu x\in X$ and $p\in \N$ given in unary}{$\reach_{2^{-p}}^\mP(\tu x, Y)$}
\end{problems}

\begin{proposition}\label{propRRLDLhard}
	$\operatorname{\bf{Reach}}$ $\operatorname{\bf{Robust}}$ $\manonclasslighttanh_\textit{light}$ is $\Pspace$-hard. 
\end{proposition}

We reduce from the problem of knowing whether a polynomial space TM reaches an accepting state. The reduction maps the set of possible configurations to the domain of the dynamical system and the transition function into the dynamic. 

\begin{proof}[Proof of Proposition \ref{propRRLDLhard}]
	Let $M = (Q, \Sigma, \Sigma', \delta_M, q_0, Q_{\accept}, Q_{\reject})$ be a $\Pspace$ TM (working in space $S$), working on the rationals, and $\tu x_M, \tu y_M \in \Q$. We construct a discrete-time dynamical system $\mP = (X, g)$ such that:
	\begin{itemize}
		\item $X = [0,|Q|] \times [0,1]^2$, a closed rational box containing the encodings of all configurations of $M$;
		\item $g = \bar{\Next}$, the function of Lemma \ref{lemNext} with parameter $m = q(n)$ and our convention on the halting states;
		\item $\tu x = \encodageconfiguration(C_0)$, $C_0$ the initial configuration of $M$ on the word $u$;
		\item $Y = B(\encodageconfiguration(C_{\text{accept}}), 4^{-(S+c)})$, a closed rational box for the sup norm;
		\item $p=q(n)$ in unary, $q$ the polynomial of Corollary \ref{corPolyRob} and $n=|u|$.
	\end{itemize}
	
	$\overline{\Next}$ is total and Lipschitz on $X$. By Lemma \ref{lemAutoCor}, the trajectory never visits non reachable configurations, thus, invalid points are never reached.
	
	By Lemma \ref{lemNext}, $g \in \mathbb{LDL}^\circ_\text{light}$ and the reduction is computable in polynomial time. By Corollary \ref{corPolyRob}, $\mP$ is polynomially robust and $\reach_{2^{-p}}^\mP(\tu x, Y) = 1$ iff $M$ accepts $u$.
	We conclude.
\end{proof}

The lower polynomial complexity bound remains true under robustness hypothesis, while undecidability does not. It comes from the two following lemmas and Corollary \ref{corPolyRob}. 
We use them in the proofs of Proposition \ref{propRRLDLhard}, Lemma \ref{DRRLHard} and Theorem \ref{CRRLhard}.

\begin{lemma}[Gap]\label{lemGap}
	Let $M$ be a Turing machine in space $S$. Let $C$ and $C'$ be two distinct configurations of $M$. Then, $\left| \encodageconfiguration(C) - \encodageconfiguration(C')\right| \geq 4^{-(S+1)} $.
\end{lemma}

	\begin{remark}\label{explicationGap}
	Lemma \ref{lemGap} guarentees that dictinct configurations, for a machine in space $S$, lie on the grid $4^{-S}\Z$ and are always separated in this grid.
\end{remark}

	\begin{proof}[Proof of Lemma \ref{lemGap}]
	Let $C = (q, l, r)$ and $C' = (q', l', r')$ be two distinct configurations.  Assume $r \neq r' $ and, without loss of generality, we assume the first index where they differ is position $i\leq S$. 
	
	We have
	$|\encodagemot(r)-\encodagemot(r')| = |r_i-r'_i|\cdot4^{-(i+1)}  + \sum_{j>i} (r_j - r'_j) 4^{-(j+1)}$. 
	
	If $r_i$ and $r'_i$ are both not $B$, the blank symbol, then, by definition of the encoding, they are in $\{1, 3\}$, so $|r_i - r'_i| = 2$. Since the numbers are all at most $3$, we have that $|\sum_{j>i} (r_j - r'_j) 4^{-(j+1)}| \leq 3 \sum_{j>i} 4^{-(j+1)} = 4^{-(i+1)} $. Hence, $|\encodagemot(r) - \encodagemot(r')| \geq 2\cdot4^{-(i+1)} - 4^{-(i+1)} = 4^{-(i+1)} \geq 4^{-(S+1)}$.
	
	If $r_i$ (or $r'_i$) is the blank symbol, then, for all $j\geq i$, $r_j=0$, thus, $|\encodagemot(r) - \encodagemot(r')| \geq r'_i\cdot 4^{-(i+1)} \geq 4^{-(i+1)}$. The proof is the same for $r'_i=0$.
\end{proof}

\begin{remark}
	The proof only uses that the blank occurs at the end of words and that the two other digits are $1$ and $3$. It would fail if all four digits were allowed, since $0.0333\ldots = 0.1$ in base $4$. Other pairs of digits would work as well, with a smaller constant: the choice of $1$ and $3$ is dictated by the width of the plateaus of $\sigma$ in Subsection \ref{subsecEnco}, not by this lemma.
\end{remark}

\begin{lemma}[Auto-correction]\label{lemAutoCor}
	Let $M$ be a Turing machine in space $S$, halting in time $T$, $Q$ its states' set, and $X$ the configurations' space. There exists a constant $c$, independent from the input, such that, for all valid configuration $C$, for all $z>0$ such that $\left|z - \encodageconfiguration(C)\right| \leq 4^{-c}$, we have:  
	\begin{enumerate}
	\item (one step) $\Next$ selects at $z$ the same branch as $M$ at $C$ (same state, scanned symbol, written symbol and move), and $\left|\Next(z) - \encodageconfiguration(\delta_M(C))\right| \leq 4\cdot\left|z - \encodageconfiguration(C)\right| + 2^{-m}$;
	\item (along the run and after halting) For $t\in\N$. If $t\leq T$, and if $(T+1) \cdot 4^S \cdot (\epsilon+2^{-m)} \leq 4^{-(S+c)}$, for all $\epsilon$-perturbed trajectory starting from $\encodageconfiguration(C_0)$, $\left| \bar{\Next}(C_t) - \encodageconfiguration(C_t)\right| \leq (t+1) \cdot 4^S\cdot (\epsilon+2^{-m})$. If $t>T$, $\left| \bar{\Next}(C_t) - \encodageconfiguration(C_t)\right| \leq \max\left(\bar{\Next}(C_t) - \encodageconfiguration(C_t), \frac{4}{3}(\epsilon+2^{-m})\right)$.
	\end{enumerate}
\end{lemma}

\begin{proof}[Proof of Lemma \ref{lemAutoCor}]
	
	\begin{enumerate}
		
		\item Write $\encodageconfiguration(C) = (q,l,r)$ and $z = (\tilde q, \tilde l, \tilde r)$. As $\left| z - \encodageconfiguration(C)\right| \leq 4^{-c} \leq 1/16$, we have $|4\tilde r - 4r| \leq 1/4$, so $4\tilde r$ lies  in the same plateau of $\sigma$ as $4r$, and $\tilde q$ is within $1/4$ of the integer $q$. Every $\sigma$, $\sendtanh$	and $\TTtanh $ in the expression of $\Next$ therefore takes the same branch at $z$ as at $\encodageconfiguration(C)$, up	to $2^{-m}$ (with the internal precision $2^{m+c}$ of Lemma \ref{lemGammaWordPtime}). On that branch, by Equation (\ref{textaremplacer}), $l'$ and $r'$ are affine in $(l,r)$ with coefficients $1/4$ or $4$, plus terms $\sigma(\cdot)$ which are exact digits up to $2^{-m}$. Hence $\left| \Next(z) - 	\Next(\encodageconfiguration(C))\right| \le 4\left| z - \encodageconfiguration(C)\right| + 2^{-m}$, and Lemma \ref{lemGammaWordPtime} concludes.
		
		\item 
		Write $e^l_t$ and $e^r_t$ for the errors on the two tape components, and $h(t)$ for the position of the head at time $t$, so that $|h(t) - h(s)| \le S$. By Equations (\ref{textaremplacer}), a right move divides $e^l$ by $4$ and multiplies $e^r$ by $4$; a left move does the opposite; each step adds at most $\epsilon + 2^{-m}$ to both. By induction on $t$, as long as the bound stays below $4^{-c}$ so that item~1 applies,
		\[
		e^r_t \le \sum_{s \le t} (\epsilon + 2^{-m})\,4^{h(t)-h(s)}
		\qquad\text{and}\qquad
		e^l_t \le \sum_{s \le t} (\epsilon + 2^{-m})\,4^{h(s)-h(t)} ,
		\]
		both at most $(t+1)\,4^{S}\,(\epsilon + 2^{-m})$. The state component is reset within $2^{-m}$ of an integer at every step by $\sendtanh$.
		
		For $t \ge T$, the convention gives $e_{t+1} \le e_t/4 + \epsilon + 2^{-m}$, and $E \mapsto E/4 + \epsilon + 2^{-m}$ maps $[0,E]$ into itself as soon as $E \ge \tfrac{4}{3}(\epsilon + 2^{-m})$.
		
	\end{enumerate}		
\end{proof}

Note that,  if $M$ is an oracle machine, the step at $C$ is a query on a word $w$, and $\left|z - \encodageconfiguration(C)\right| \leq 4^{-(S+2)}$, then the oracle component of $\bar{\Next}(z)$ is exactly $\encodagemot(\varphi(w))$.
\begin{wrapfigure}{r}{0.4\textwidth}
	\vspace{-.4cm}
	\begin{tikzpicture}[scale=.6]
		\draw[->] (0,0) -- (10,0);
		\draw (10,0) node [anchor=north]  {$\R$};
		\draw[-, dashed, color=lightgray] (2,0) -- (2,-1.3);
		\draw[-, dashed, color=lightgray] (8,0) -- (8,-1.3);
		\draw[-, dashed, color=lightgray] (7,0) -- (7,1.5);
		\draw[-, dashed, color=lightgray] (9,0) -- (9,1.5);
		\draw[-, dashed, color=lightgray] (1.7,0) -- (1.7,1);
		\draw[-, dashed, color=lightgray] (2.3,0) -- (2.3,1);
		\draw[-] (2,-.1) -- (2,.1) node [anchor=north][inner sep=0.3cm]   {$\gamma(C_i)$};
		\draw[-] (8,-.1) -- (8,.1) node [anchor=north][inner sep=0.3cm]   {$\gamma(C_{i+1})$};
		\draw[<->] (2,-1.3) -- (8,-1.3) node[midway] [anchor=south][inner sep=0.1cm] {$\geq 4^{-(S+1)}$} ;
		\draw (1,0) node   {\huge{$[$}};
		\draw (3,0) node   {\huge{$]$}};	
		\draw[<->] (7,1.5) -- (9,1.5) node[midway] [anchor=south][inner sep=0.1cm] {$\leq 4^{-(S+c)}$} ;	
		\draw[<->] (1.7,1) -- (2.3,1) node[midway] [anchor=south][inner sep=0.2cm] {\tiny $\leq 2^{-q(n)}$} ;	
		\draw (1.7,0) node   {\small{$[$}};
		\draw (2.3,0) node   {\small{$]$}};	
		\draw[thick,->,MidnightBlue] (2,0) to [bend left=45] node[midway, above] {$\delta_M(C)$} (8, 0);
		\draw (7,0) node   {\huge{$[$}};
		\draw (9,0) node   {\huge{$]$}};
		\draw (7.7,0) node   {\small{$[$}};
		\draw (8.3,0) node   {\small{$]$}};			
	\end{tikzpicture}
	\vspace{-.9cm}
\end{wrapfigure}

Hence, the perturbed trajectory, which remains close to the sequence of configurations of $M$, never approaches invalid points.
Intuitively, the error at each step is absorbed by the gap between the encodings, making the system robust:

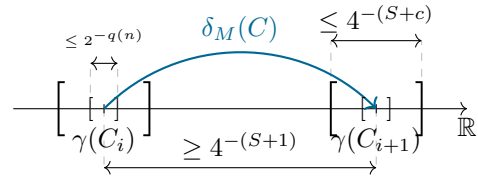
\begin{corollary}[Polynomial robustness of $H$]\label{corPolyRob}
	Take $q(n) = 9\,S(n) + 2c + \lceil \log_2 |Q| \rceil + 3$ and $m = q(n)$. For $\epsilon = 2^{-q(n)}$, $(T_{\max}+1)\,4^{S}\,(\epsilon + 2^{-m}) \le 4^{-(S+c)}$, so, by Lemma \ref{lemAutoCor}, every $\epsilon$-perturbed trajectory from $\encodageconfiguration(C_0)$ stays, at all times, within $4^{-(S+c)}$ of the encoding of the current configuration of $M$. Since distinct configurations are at distance at least $4^{-(S+1)}$ (Lemma \ref{lemGap}), and $c \ge 2$, a perturbed trajectory enters $Y = B(\encodageconfiguration(C_{\text{accept}}), 4^{-(S+c)})$ iff $M$ accepts. Hence, $\reach^\mP_{2^{-q(n)}}(x,Y) = \reach^\mP(x,Y)$: $\mP$ is polynomially robust.
\end{corollary}

\begin{proof}[Proof of Corollary \ref{corPolyRob}]
	Since $T_{\max} + 1 \le 2|Q|\,S\,3^{2S} \le 2^{5S + \lceil\log_2|Q|\rceil + 2}$ and $\epsilon + 2^{-m} = 2^{1-q(n)}$, the choice of $q$ gives $(T_{\max}+1)\,4^{S}\,(\epsilon + 2^{-m}) \le 4^{-(S+c)}$. Lemma \ref{lemAutoCor} bounds the distance to the encoding of the current configuration by $4^{-(S+c)}$ at all times. If $M$ accepts, the unperturbed trajectory ends within $4^{-(S+c)}$ of $\encodageconfiguration(C_{\text{accept}})$; if not, every perturbed trajectory stays within $4^{-(S+c)}$ of configurations distinct from $C_{\text{accept}}$, hence at distance more than $4^{-(S+1)} - 4^{-(S+c)} > 4^{-(S+c)}$ from it since $c \ge 2$.
\end{proof}

The consequence of these lemmas is that, in unbounded space $S$, $4^{-(S+c)}$ converges to $0$: undecidability does not survive under robustness hypothesis, while the lower complexity bound does. It marks the border between the main complexity results of this paper and classical undecidability theory.
As this problem is in $\Pspace$ (\cite{BlancBournezJoC}):

\begin{corollary}\label{RRLightComp}
	$\operatorname{\bf{Reach}}$ $\operatorname{\bf{Robust}}$ $\manonclasslighttanh_\textit{light}$ is $\Pspace$-complete. 
\end{corollary}
We have a first $\Pspace$-completeness result, using the classical notion of reduction. 
	If we had considered $\manonclasslighttanh_\textit{light}\backslash \{\tanh\}$, we would be in the situation where the dynamic function is a PAM. The proof of Proposition \ref{propRRLDLhard} also covers this case.
We now prove this property but with a stronger notion of reduction and for more general domains, using oracle TMs.

\section{$\Pspace^2$-$\leq_m^2$-completeness for the robust $\textbf{Reach}$ with oracle TMs}\label{secPSComplete}

We impose our dynamic functions to be Lipschitz. We are going to use the notion of reduction given in Definition \ref{defTopoRed}. 
\begin{definition}[Oracle as a real function]
	Let $\varphi \in \bf{Reg}$ and $S \in \N$. The function $\Phi_S : [0,1] \to [0,1]$ is the piecewise affine function which is constant, equal to $\encodagemot(\varphi(w))$, on the interval $\left[\encodagemot(w) - 4^{-(S+2)},\encodagemot(w) + 4^{-(S+2)}\right]$ for every word $w$ of length at most $S$, and affine between two consecutive such intervals.
\end{definition}

Thus, one oracle step is also auto-correcting (see Lemma \ref{lem63}).
We write $\textbf{PAM}$ for the algebra of PAMs and the operators under which it is closed. We denote by $\mathcal{A} + \Phi$ the algebra where $\mathcal{A}$ is an algebra and $\Phi$ is added as a base function.
It gives us a strong stability result, avoiding too \emph{chaotic} systems.
We define the problem \textbf{DRRL}:

\begin{problems}
	{Discrete Reach Robust Lipschitz (DRRL)}{A discrete-time polynomially robust dynamical system $\mathcal{H} = (X,g)$, $X$ a closed rational box, $g : X \rightarrow X\in\manonclasslighttanh+\Phi$, $g$ Lipschitz, $p\in\N$ given in unary, $\tu x, \tu y \in X $}{$\reach^\mathcal{H}_{2^{-p}}(\tu x, \tu y)$}
\end{problems}

The main difference with the problems considered in \cite{KawCook12} is that we do not bound the time: there is no fixed number of iterations of $g$ in $\mathcal{H}$ in the input.

\subsection{Transition functions for oracle machines}

Many authors have embedded TMs in various classes of dynamical systems to get undecidability or several hardness results.  
We present in this section how this is usually done, to point out where (intuitively) (non-)robustness issues appear.

\begin{theorem}[\cite{BlancBournezMCU22vantardise,BlancBournezMFCS2023vantardise}]\label{thEmbeddingTM}
	Let $M$ be a TM. Then, we can construct a dynamical system $(\N\times\Sigma^\omega\times\Sigma^\omega, \tu f)$, $\tu f$ a PAM with rational coefficients, simulating one transition of $M$.
\end{theorem}

A proof is in Appendix \ref{appPreuveTH}.
The transition function of a TM can be seen as a PAM (Theorem \ref{thEmbeddingTM}). 
Note that, the transition function in a(-n oracle) TM is Lipschitz, a feature the proof relies on. It comes from the fact that, from the encoding we consider, applying one transition implies that the encoded real is at most multiplied by 4 and at least divided by $\frac14$.
We extend Theorem \ref{thEmbeddingTM} to deal with real numbers. We prove we can simulate an oracle since we are dealing with an oracle TM.

\begin{proposition}
	Let $M = \left(Q, \Sigma, \Delta_M, q_{\init}, Q_{\accept}, Q_{\reject}\right)$ be an oracle TM. We denote by $\tu f$ the transition function of $M$ with oracle $\phi$. We have that $\tu f$ is in $\textbf{PAM} + \Phi$.
\end{proposition}

Since $g\in \manonclasslighttanh + \Phi$, it is computable in second-order polynomial time.
Hence, with Theorem \ref{thdirectionunpbis} and Corollary \ref{coroTransPtime},  we have that 
\textbf{DRRL} $\in\Pspace^2$.

\subsection{Discrete-time systems}

As in \cite{CSL24}, we restrain our dynamic functions to be Lipschitz. 

\begin{theorem}[\cite{CSL24}, {\cite[Theorem~45]{BlancBournezJoC}}]\label{thmPSmember}
	$\operatorname{\bf{DRRL}}\in\Pspace^2$. 
\end{theorem}

The proof of Theorem \ref{thmPSmember} uses the same idea as the proof of Savitch's theorem \cite{Sip97}, where we reason on the reachability in a graph. We properly discretise the domain of the dynamical system, using the error we allow, and see the ``cubes'' as the vertices of a graph. 

\paragraph*{Technical Lemma: Completeness result}

\begin{lemma}\label{DRRLHard}
	$\operatorname{\bf{DRRL}}$ is $\Pspace^2-\leq_m^2$-hard.
\end{lemma}
	
The proof consists in reducing (Definition \ref{defTopoRed}) \textbf{Basic $\Pspace^2$} to \textbf{DRRL}:
The function $\mathfrak{t}$ maps the initial configuration of $M^\phi$ to the starting point of the trajectory. And the function $\mathfrak{s}$ encodes in the dynamical system the execution of the machine.
More formally, we first state and use the following lemma:

\begin{lemma}\label{lem63}
	Let $\varphi \in \mathbf{Reg}$ and $S \in \N$. Then: 
	\begin{itemize}
		\item (i) $\Phi_S$ is Lipschitz, with constant at most $8\cdot 4^{S}$; 
		\item (ii) $\Phi_S$ is computable in second-order polynomial time; 
		\item (iii) $\Phi_S(\tilde z) = \encodagemot(\varphi(w))$ whenever $|w| \le S$ and $|\tilde z - \encodagemot(w)| \le 4^{-(S+2)}$.
	\end{itemize}
\end{lemma}
\begin{proof}[Proof of Lemma \ref{lem63}]
	(iii) is the definition. For (i), by Lemma \ref{lemGap} two plateaus are separated by an interval of length at least $2\cdot 4^{-(S+2)}$, on which $\Phi_S$ varies by at most $1$. For (ii), knowing $z$ at precision $2^{-n}/(8\cdot 4^{S})$ determines the one or two words whose plateaus surround $z$; query $\varphi$ on them and interpolate.
\end{proof}

\begin{proof}[Proof of Lemma \ref{DRRLHard}]
	We reduce \textbf{Basic $\Pspace^2$} to \textbf{DRRL}.
	Let $M, \mu, \phi$ and $u$ be the inputs of \textbf{Basic $\Pspace^2$}. Without loss of generality, we assume $M^\phi = \left( Q, \Sigma, \tu f, q_{\init}, Q_{\accept}, Q_\reject\right)$ has a unique accepting state $q_{\accept}$ with  $0^{\mu(|u|)}$ being the content of the tape. 
	We construct the inputs of $\textbf{DRRL}$ from the inputs of \textbf{Basic $\Pspace^2$} and define the functions $\mathfrak{s}$ and $\mathfrak{t}$ from Definition \ref{defTopoRed}: 
	\begin{itemize}
		\item For $C_0$ the initial configuration of $M^\phi$ on $u$, we define $\tu x = \encodageconfiguration(C_0)$. 
		\item For $C_{f_a}$ the (unique) final accepting configuration of $M$,  $\tu y = \encodageconfiguration(C_{f_a})$.
		\item $\mathfrak{s}: <M, \mu, \phi> \rightarrow X,g, p,\tu y$, where  $X$ is the set of all the real numbers encoded by possible configurations of $M^\phi$, 
		$S = \mu(|u|)$ is the space bound, $p = q(|u|)$ with $q$ the polynomial of Corollary \ref{corPolyRob}, and $g$ is $\bar{\Next}$ with parameter $m = 2\cdot q(|u|)$, so that the drift of $\bar{\Next}$ itself along the run, at most $(T_{\max}+1)\,4^{S}\,2^{-m}$, stays below the radius $2^{-p}$ of the target ball.
		\item $\mathfrak{t}$ sends the input $u$ to $\encodageconfiguration(C_0)$.
	\end{itemize}
	The function $\mathfrak{s}$  builds the \textbf{DRRL} instance; its dynamic $g$ is Lipschitz and computable in second-order polynomial time by Lemma \ref{lemNext} and Corollary \ref{coroTransPtime} and it is Lipschitz 
	The function $\mathfrak{t}$ maps the input $u$ to the starting point of the trajectory in $X$. 
	It is computable in polynomial time by Lemma \ref{lemGammaWordPtime}.
	The system  $(X,g)$ is polynomially robust by Corollary \ref{corPolyRob}, and we have $(\frak{s}(\phi)) \circ \frak{t}(\phi) (u)$ corresponds to ``$M^\phi$ accept $u$''.
	Thus, we have $\textbf{DRRL}$ is $\Pspace^2$-$\leq_m^2$-hard.
\end{proof}

\begin{corollary}\label{DRRLComp}
	$\operatorname{\bf{DRRL}}$ is $\Pspace^2$-$\leq_m^2$-complete.
\end{corollary}

	\subsection{Continuous-time systems}
	
	The goal of this subsection is to extend Lemma \ref{DRRLHard}. It relies on previous results from \cite{ICALP24}. The main new results here are Lemma \ref{lemAmeSin} and Theorem \ref{CRRLhard}.

	\begin{definition}[Continuous-time Dynamical Systems]\label{defCTDS}
		A continuous-time dynamical system is a tuple $(X, g_h)$, with $X$ a set and $g_h : \R\times X \rightarrow X$ such that $g_h$ is the solution of
			$g_h(0,\tu x) = y(\tu x)$ (some initial condition depending only on $\tu x$) and $\frac{\partial g_h(t,\tu x)}{\partial t} = h\left( t, g_h(t, \tu x) \right)$.    
	\end{definition}

 Kawamura has shown in \cite{kawamura2009lipschitz} that solving Lipschitz ODEs is PSPACE-complete. Our results differ in three respects. First, we decide a reachability relation rather than compute a solution, with no time bound given as part of the input. Second, we do not restrict ourselves to compact domains. Third, our systems are assumed to be robust, an assumption under which decidability is already non trivial \cite{asarin01perturbed, CSL24}, and our lower bound holds for the second-order reductions of \cite{KawCook12}. 

We can extend our previous result to \emph{continuous}-time dynamical systems:

	\begin{problems}
	{Continuous Reach Robust Lipschitz (CRRL)}{A continuous-time polynomially robust dynamical system $\mP_C = (X_C,g_C)$, with $X_C\subseteq\R$ a closed rational box, $g_C : \R \times X_C \rightarrow X_C\in\contClass+\Phi$, $g_C$ Lipschitz, $p\in\N$ in unary, $\tu x, \tu y \in X_C $}{$\reach^{\mP_C}_{2^{-p}}(\tu x, \tu y)$}
\end{problems}

	\paragraph*{The Simulation Trick: a Polynomial Improvement}

To prove the $\Pspace^2$-hardness of this problem, we reduce from \textbf{DRRL}. 
We first prove that we can simulate a discrete ODE with a continuous ODE. 
We rely on a simulation of a discrete dynamic by a continuous ODE attributed to Branicky (\cite{Bra95}). 
The core idea of this method is, starting from a discrete dynamic,
to construct a continuous ODE 
whose (smooth) dynamic goes through all the points of the discrete dynamic.
We improve this method for our context (Lemma \ref{lemAmeSin}).

	\begin{definition}[Ideal iteration trick \cite{Bra95}] \label{defBranicky}
		Consider the following discrete ODE, given by functions $\tu g$ and $\tu u$: $\tu f(0,\tu x) = \tu g(\tu x); \frac{\delta \tu f}{\delta t}(t,\tu x) = \tu u \left(\tu f(t,\tu x),t,\tu x\right).$		
		Then, let $\tu G(\tu v, t, \tu x) = \tu u(\tu v,t, \tu x) + \tu v$, 
		and consider the (continuous) ODE: 
		\vspace{-.2cm}
		\begin{equation} \label{code}
				\tu y_1' = \mathfrak{c} \cdot \left(\tu G(r(\tu y_2),r(t),\tu x) - \tu y_1 \right)^3 \theta\left(\sin(2\pi t)\right) \quad\text{and}\quad
				\tu y_2' = \mathfrak{c} \cdot \left(r(\tu y_1) - \tu y_2 \right)^3 \theta\left(- \sin(2\pi t)\right)
		\end{equation}
		where $\tu y_1(0,\tu x) = \tu y_2(0,\tu x) = \tu g(\tu x)$,  $\mathfrak{c}$ is a constant, $\theta(x)=0$ if $x\leq 0$ and $\theta(x) >0$ if $x>0$.  Here, $r$ is a rounding function: we mean, by construction, $G$ preserves the integers, and $r$ is a function that maps a real value close to some integer to this integer: assume, say,  that for $z \in \left[n-\frac14,n+\frac12\right]$, $r(z)=n$, for any integer $n \in \Z$. 
	\end{definition}

	We write $r(\tu v)$ for applying the function  $r: \R \to \R$ componentwise on the vector $\tu v$.
We do not need to specify what $r(z)$ values for a $z$ outside such interval: the following reasoning remains correct, whatever it is.
The solution of the continuous ODE in Definition \ref{defBranicky} emulates continuously the associated discrete ODE.
Intuitively, for $f$ the solution of the discrete ODE and $g$ the solution of the constructed continuous ODE, for any $\tu x, \tu z \in\R$, if $f(\tu x) = \tu z$ then $g(\tu x) = \tu z$. 
More details about this proposition and an improvement (from \cite{ICALP24}) are in Appendix \ref{detailBra}.
We actually improve the latter result, by expressing the system (\ref{code}) without using a sine. The point is to replace $t\mapsto \sin(2\pi t) $ by function $\tu y$ such that: 

\begin{equation}\label{eqSinus}
	\frac{\partial \tu y}{\partial t} (t) = 2 \pi \cdot \tu z(t) \quad\quad\text{ and }\quad\quad \frac{\partial \tu z}{\partial t} (t) = - 2 \pi \cdot \tu y(t)
\end{equation}
where $\tu y(0) = 0$ and $\tu z(0) = 1$.
This is a \emph{polynomial} system of ODEs, their solutions are 1-periodic. It enables us to avoid the (direct) use of sine and cosine functions in the simulation, and still have an expression for periodic functions. It also simplifies the algebra $\contClasslim$: the cosine as a base function is no longer necessary.
The main point is to make sure the errors remain additive, so the simulation stays in polynomial time.

\begin{lemma}\label{lemAmeSin}
	Suppose that, in \eqref{code}, we replace function $\theta(z)$ and function $r(z)$ by some suitable approximations: we take $\theta(x)=\relu(x)$, $\theta_{\epsilon'}(x)$, $r_{\epsilon'(z)}$ such that $\theta(z)=_{\epsilon'}\theta_{\epsilon'}(z)$ and $r(x)=_{\epsilon'}r_{\epsilon'}(x)$, take constant $\mathfrak{c}$ big enough, and replacing $t \mapsto \sin(2\pi t)$ by Equation \ref{eqSinus}, the solution of the obtained ODE continuously simulates the associated discrete ODE, with error at most $\epsilon$ if $\epsilon'$ is taken sufficiently small.
\end{lemma}

The proof is an induction on the intervals of the domain. The main difficulty of the proof is to ensure that the computation errors remain additive (not multiplicative), to guarentee polynomial space. 

	As in Proposition \ref{the:manon:trick}, to guarantee an error of $\epsilon=2^{-n}$, it is sufficient to take $\epsilon'=2^{-q(n)}$  and $\theta_{\epsilon'}(x)=\relutanh(2^{q(n)},x)$ for some polynomial $q$. 
Note that Equation (\ref{eqSinus}) is a robust continuous ODE (see Definition \ref{defRobODE}).
To prove this lemma, we first state the following one:

\begin{lemma}[{Improved error analysis of ODE \eqref{eq:de:branicky},   \cite[Lemma 4.5]{TheseRiccardo} \cite[5.2]{BGDPRiccardo2022}}]  \label{lem:quatre:cinq}
	Consider a point $b \in \R$, some $\gamma>0$ some reals $t_0<t_1$, and a function $\phi: \R \rightarrow \R$ with the property that $\phi(t) \geq 0$ for all $t \geq t_0$ and $\int_{t_0}^{t_1} \phi(t) d t>0$. Let $\rho, \delta \geq 0$ and let $\bar{b}, E: \mathbb{R} \rightarrow \mathbb{R}$ be functions such that that $|\bar{b}(t)-b| \leq \rho$ and $|E(t)| \leq \delta$ for all $t \geq t_0$. Then the IVP defined by
	\begin{equation}\label{eqLemError}
		z^{\prime}=\mathfrak{c}\left(\bar{b}(t)-z\right)^3 \phi(t)+E(t)
	\end{equation}
	with the initial condition $z\left(t_0\right)=\bar{z}_0$, 
	where $\gamma>0$ and $\mathfrak{c} \ge \frac{1}{2 \gamma^2 \int_{t_0}^{t_1} \phi(t) dt}$ satisfies
	\begin{enumerate}
		\item  $\left|z\left(t_1\right)-b\right|<\rho+\gamma+\delta\left(t_1-t_0\right)$, independently of the initial condition $\bar{z}_0 \in \mathbb{R}$
		\item $\min \left(\bar{z}_0, b-\rho\right)-\delta\left(t_1-t_0\right) \leq z(t) \leq \max \left(\bar{z}_0, b+\rho\right)+\delta\left(t_1-t_0\right)$ for all $t \in\left[t_0, t_1\right]$.
	\end{enumerate}
\end{lemma}

We notice that $\theta_{\epsilon'}$ is never exactly $0$, because $\tanh$ is never exactly $0$: it is precisely what the function $E$ in Equation \ref{eqLemError} implies. The function $r_{\epsilon'}$ corresponds to a continuous approximation (at precision $\epsilon'$) of a simple function. An explicit construction can be found in \cite{BGDPRiccardo2022}.

\begin{lemma}[Step-by-step reinitialisation]\label{lemInterStep} 
	As $r_\epsilon'$ is almost constant on $[k \cdot \eta - \frac{\eta}{4}, k \cdot\eta + \frac{\eta}{2}]$, for all $k\in\N$, every error smaller than $\frac{\eta}{4}$ is adjusted back to $\epsilon'$ at the next step.
\end{lemma}

\begin{proof}[Proof of Lemma \ref{lemInterStep}]
	By definition of $r_{\epsilon'}$ in Proposition \ref{the:manon:trick}, for $z\in X$, if $\left|z - k \cdot \eta\right| \leq \frac{\eta}{4}$, then $r_{\epsilon'}(z) =_{\epsilon'} k$. 
	
\end{proof}

\begin{remark}
	By Lemma \ref{lemGap}, Remark \ref{explicationGap}, the condition $\epsilon'\leq\frac{\eta}{4} $ becomes $\epsilon' < 4^{-(S+1)}$, which is exactly the condition needed for $q$ in the proof of Lemma \ref{lemAutoCor} and Corollary \ref{corPolyRob}.
\end{remark}

\begin{proof}[Proof of Lemma \ref{lemAmeSin}]
	The proof is similar to the proof of Proposition \ref{the:manon:trick}, except we have to take into account the additional error of the computation of Equation (\ref{eqSinus}). The key here is also to make sure the errors propagate additively (and not multiplicatively): rounding function corrects errors or order $\epsilon$ whenever its argument is at a distance less than $\frac14\eta$ of some $n\eta$ (Lemma \ref{lemInterStep}). We introduce some error $\epsilon'$ at every step; but the latter is corrected at the next step). 
	The involved constant $\mathfrak{c}$ is of order $2^{n}$,
	but it only multiplies $\phi(t)$ in Equation \ref{eqLemError} and enters the bound of Lemma \ref{lem:quatre:cinq} through $\int\phi$. A larger $\mathfrak{c}$ makes $z$ converge faster to $\bar b$ and leaves the error terms $\rho$, $\gamma$ and $\eta(t_1-t_0)$ unchanged. Since the cubic term is contracting towards $\bar b$, a large $\mathfrak{c}$ damps perturbations rather than amplifying them, so the numerical stability of Definition \ref{defRobODE} depends on $\epsilon'$ only, and $\mathfrak{c}$ is written with $n$ doublings, so Lemma \ref{lemNewBra} is unaffected.
	We use the notations of Definition \ref{defBranicky}.
	
	We prove, by induction, the following property: for all $n\in\N$, $\tu y_1\left(n,\tu x\right) =_\epsilon \tu f\left(n,\tu x\right)$ and $\tu y_2\left(n,\tu x\right) =_\epsilon \tu f\left(n,\tu x\right)$. Also, for all $t \in \left[n,n+\frac12\right]$, we have  $\tu y_1\left(t+\frac12,\tu x\right) =_{\frac{\eta}{4}}  \tu f\left(n,\tu x\right)$ and $\tu y_2\left(t,\tu x\right) =_{\frac{\eta}{4}} \tu f\left(n,\tu x\right)$.
	
	By Lemma \ref{lemInterStep}, the induction is \emph{not} on the sum of errors. The number of steps does not influence the analysis. Only the preservation of the error at each step does.
	
	\underline{Initialisation:}	
	
	For $n=0$, $\tu f(0,\tu x) = \tu y_{1}(0,\tu x) = \tu y_{2}(0,\tu x) = \tu g(\tu x)$.
	\begin{itemize}
		\item For $t \in [0,1/2]$, $r(\tu y_{2})$ is kept close to a constant value $\tu g(\tu x)$, when an error less than $\epsilon'$, if we choose $\epsilon' < \frac14 \eta$. Meanwhile, $r(t)$ is also at a value close to $n$ with an error lesser than $\epsilon'$. 
		We have $\theta_{\epsilon'}\left(-\tu{y(t)}\right) =_{\epsilon'} 0$, thus $\tu{y_2}(t,\tu x) =_{\epsilon'}  0 $, then $y_2(t, \tu x)$ remains close to $y_2(0, \tu x)$ on this interval.
		Here the hypothesis that $\tu G$ is almost constant around  $\N\eta$ means that its value is guaranteed to be at $\epsilon'$ from $\tu G\left(\tu g(\tu x),0,\tu x\right)$ on the $\left[0, \frac{1}{2}\right] $. 
		
		By Lemma \ref{lem:quatre:cinq}, $\tu y_1(t,\tu x)$ approach $\tu G\left(\tu g(\tu x),0,\tu x\right)=\tu f(1,\tu x)$ on this interval, with an error of order	$\epsilon' + \epsilon' + \epsilon' + \frac12 \epsilon'$. 
		Thus, $\tu y_1(\frac{1}{2},\tu x) =_{\epsilon/2}  \tu f(1,\tu x)$, if we choose $\frac72 \epsilon' < \epsilon/2$. 
		
		\item At $t=\frac{1}{2}$, $\tu y_1$ will hence have simulated one step of the associated discrete ODE, with an error less than $\epsilon/2$, 
		and $\tu y_{2}$ will be close to $ \tu g(\tu x)$ with an error less than $\epsilon' < \epsilon/2$.

		\item For $t \in \left[\frac{1}{2}, 1\right]$ the roles of $y_1$ and $y_2$ are inverted, and the analysis is analogous.
	\end{itemize}
	
	\underline{Induction:}
	\begin{itemize}
		\item For $t \in [n,n+1/2]$, we have $\theta(-\tu y (t))=_{\epsilon'} 0$, and hence $\tu y_2'(t) =_{\epsilon'} 0$, so $\tu y_2$ is kept almost fixed (with an error less than $\frac12 \epsilon'$). In the same time $\tu y_1$ approaches $G(r(\tu y_{2}(t), r(t), \tu x)=\tu f(n+1, \tu x)$ by Lemma \ref{lem:quatre:cinq}, with some new error of order less than $\frac72 \epsilon' < \epsilon/2$. 
		
		\item For $t \in \left[n+\frac{1}{2}, n+1\right]$ the roles of $\tu y_1$ and $\tu y_2$ are exchanged: $\tu y_1'(t, \tu x) =_{\epsilon'} 0$, so $\tu y_1$ is kept almost fixed, with a new error less than $\frac12 \epsilon'$.  In the same time $\tu y_2$ approaches $r(\tu y_{1}(t))=\tu f(n+1, \tu x)$ by Lemma \ref{lem:quatre:cinq}, with some new error	of order less than $\frac72 \epsilon' < \epsilon/2$. 
	\end{itemize}
	Consequently, we get the property at rank $n+1$. 
\end{proof}

Since polynomials are stable under composition, we have:

\begin{lemma}\label{lemNewBra}
	The simulation of Lemma \ref{lemAmeSin} is computable in polynomial time.
\end{lemma}

Hence, by Theorem \ref{thmPSmember} and Lemma \ref{lemNewBra}:

\begin{corollary}
	$\operatorname{\bf{CRRL}}\in\Pspace^2$. 
\end{corollary}

A side-effect of Lemma \ref{lemAmeSin} is that the cosine function is no longer necessary in the basis of the algebra $\contClass$, as it can be expressed with the other components of the algebra:

\begin{theorem}[Main theorem I] \label{thNewRCD} 
	$\contClasslim \cap \R^{\N^{d} \times \R^{d'}}= \FPspace \cap \R^{\N^{d} \times \R^{d}}$, where $\contClasslim = [0, 1,\pi_i^k, +,$ $ -, \times, \tanh, \pi, \frac{x}{2}, \frac{x}{3};$ $\composition,$ $ \robODE, \MANONlim].$
\end{theorem}

\medskip
\noindent
\textbf{\sf \Large Completeness result.}
We can now state and prove the following theorem:

\begin{theorem}[Main theorem II]\label{CRRLhard}
	$\operatorname{\bf{CRRL}}$ is $\Pspace^2$-$\leq_m^2$-hard.
\end{theorem}

We reduce \textbf{DRRL} to \textbf{CRRL}, where $\mathfrak{s}$ encodes the simulation of Lemma \ref{lemAmeSin} and $\mathfrak{t}$ the successor function (for the discrete ODE): 

\begin{proof}[Proof of Theorem \ref{CRRLhard}]
	We reduce from \textbf{DRRL}:
	\begin{itemize}
		\item $\frak{s}: \mP=(X,g) \rightarrow \mP_C=(X_C, g_C)$ doing the adaptation of the Branicky trick described by Proposition \ref{the:manon:trick}: with $X=X_C$ and $g_C$ such that:
		\begin{equation*}
			\left\{
			\begin{aligned}
				& g_C(0, \vx) = g_C^m(0, \vx) = g(0, \vx) = \encodageconfiguration(C_0); \tu y(0) = 0; \tu z(0) = 1\\
				& \frac{\partial g_C}{\partial t}(t, \vx) = \mathfrak{c} \cdot \left( G\left( r_{\epsilon'}\left(g_C^m(t, \vx)\right), r_{\epsilon'}(t), \vx \right) - g_C(t,\vx) \right)^3 \cdot \theta_{\epsilon'}\left(\tu y(t)\right) \\
				& \frac{\partial g_C^m}{\partial t}(t, \vx) = \mathfrak{c} \cdot \left( r_{\epsilon'}\left(g_C(t, \vx)\right) - g_C^m(t, \vx)\right)^3 \cdot \theta_{\epsilon'}\left( -\tu y(t) \right)\\
				&\frac{\partial \tu y}{\partial t} (t) = 2 \pi \cdot \tu z(t) \quad \quad \text{and}\quad\quad \frac{\partial \tu z}{\partial t} (t) = - 2 \pi \cdot \tu y(t)
			\end{aligned}
			\right.
		\end{equation*}	
		where $c$ is a constant; $G: (v, t, \vx) \rightarrow \tu u(v, t, \vx) + v $ and $\tu u$ such that $\frac{\delta g}{\delta t}(t,\vx) = \tu u\left( g(t, \vx), t, \vx \right)$; $r_{\epsilon'}$ and $\theta_{\epsilon'}$ are defined in Proposition \ref{the:manon:trick}. The function $g_C^m$ has the value of a memory equation, as $y_2$ in the Definition \ref{defBranicky}.
		It is computable in polynomial time by Lemma \ref{lemNewBra}.
		\item $\frak{t}:  \gamma(w) \rightarrow \gamma(w')$ where $w'$ is the successor of $w$ in the discrete ODE.
	\end{itemize}
	$\mP_C$ is robust by Lemma \ref{lemAutoCor} and Corollary \ref{corPolyRob}.
	The equivalence is direct from Lemma \ref{lemAmeSin}.
\end{proof}

\begin{corollary}\label{CRRLComp}
	$\operatorname{\bf{CRRL}}$ is $\Pspace^2$-$\leq_m^2$-complete.
\end{corollary}

\section{Conclusion}\label{secConclu}

We proved $\Pspace$-completeness for $\textbf{Reach}$, under the framework of existing algebraic characterisations of $\Pspace$, defined in \cite{BlancBournezMFCS2023vantardise}. The proof only required the classic notion of reduction and is about classical complexity. 

We then proposed a stronger property: we managed to prove that the reachability relation is $\Pspace^2$-$\leq_m^2$-complete for robust discrete-time dynamical systems in second-order complexity. We provide those completeness results in a more general framework than in \cite{KawCook12}, as we use a stronger notion of reduction on more general domains. 
So now, we have:
\begin{center}
\begin{tabular}{|c|c|c|c|}
	\hline
	& $\textbf{Reach}$ Robust $\manonclasslighttanh_\textit{light}$ & \textbf{DRRL} & \textbf{CRRL} \\
	\hline
	-Completeness & $\Pspace$ (\textcolor{Violet}{Cor. \ref{RRLightComp}}) & $\Pspace^2$ (\textcolor{Violet}{Cor. \ref{DRRLComp}}) & $\Pspace^2$ (\textcolor{Violet}{Cor. \ref{CRRLComp}}) \\
	\hline
\end{tabular}
\end{center}

Having such completeness results allows us to have a better understanding of those complexity classes over the reals.

An interesting side-effect is also that we managed to remove the cosine function from the algebra $\contClasslim$, characterising polynomial space for functions over the reals (Theorem \ref{thNewRCD}), getting us closer to a minimal algebra to characterise polynomial space over the reals.

Many further works are possible. For example, reasoning on broader definitions of perturbations and \emph{noisy} dynamical systems to include real-world systems would be interesting. More precisely, having completeness results for classical properties (reachability, existence of an attractor, etc.) for more general types of systems would allow a deeper understanding of the upper and lower complexity bounds. 
For the reduction of the base of $\contClass$, the same construction suggests that $\pi$ could also be removed from the basis by rescaling time, we also leave this for future work.

	\newpage
	\bibliographystyle{plainurl}
	\bibliography{bibli, bib2doi}

@inproceedings{gao2012delta,
	author = {Gao, Sicun and Avigad, Jeremy and Clarke, Edmund M},
	booktitle = {Logic in Computer Science (LICS), 2012 27th Annual IEEE Symposium on},
	organization = {IEEE},
	pages = {305--314},
	title = {Delta-decidability over the reals},
	year = {2012}}

@article{BloTsi00,
	author = {Blondel, Vincent and Tsitsiklis, John},
	journal = {Automatica},
	number = {9},
	pages = {1249--1274},
	title = {A survey of computational complexity results in systems and control},
	volume = {36},
	year = {2000}}

@article{kawamura2015computational,
	author = {Kawamura, Akitoshi and M{\"u}ller, Norbert and R{\"o}snick, Carsten and Ziegler, Martin},
	journal = {Journal of Complexity},
	number = {5},
	pages = {689--714},
	publisher = {Elsevier},
	title = {Computational benefit of smoothness: Parameterized bit-complexity of numerical operators on analytic functions and Gevrey's hierarchy},
	volume = {31},
	year = {2015}}

@inproceedings{kawamura2018parameterized,
	author = {Kawamura, Akitoshi and Steinberg, Florian and Thies, Holger},
	booktitle = {International Workshop on Logic, Language, Information, and Computation},
	organization = {Springer},
	pages = {223--236},
	title = {Parameterized complexity for uniform operators on multidimensional analytic functions and ODE solving},
	year = {2018}}

@article{HKPV95,
	author = {Henzinger, Thomas A. and Kopke, Peter W. and Puri, Anuj and Varaiya, Pravin},
	journal = {Journal of Computer and System Sciences},
	month = aug,
	number = {1},
	pages = {94--124},
	preliminary = {STOC::HenzingerKPV1995},
	references = {\cite{TCS::AlurCHHHNOSY1995} \cite{TCS::AlurD1994} \cite{IEEETSE::AlurHH1996} \cite{STOC::AlurHV1993} \cite{LICS::BouajjaniES1993} \cite{IEETSE::Corbett1996} \cite{ICALP::HenzingerMP1992} \cite{IC::HenzingerNSY1994} \cite{IC::LynchV1996} \cite{LICS::VardiW1986}},
	title = {What's Decidable about Hybrid Automata?},
	volume = {57},
	year = {1998}}

@article{AMP95,
	acknowledgement = ack-nhfb,
	author = {Asarin, Eugene and Maler, Oded and Pnueli, Amir},
	bibdate = {Sat Nov 22 13:15:15 MST 1997},
	classification = {C1160 (Combinatorial mathematics); C1310 (Control system analysis and synthesis methods); C4290 (Other computer theory); C6110 (Systems analysis and programming)},
	coden = {TCSCDI},
	corpsource = {Inst. of Inf. Transmission Problems, Acad. of Sci., Moscow, Russia},
	day = {6},
	issn = {0304-3975},
	journal = {Theoretical Computer Science},
	month = feb,
	number = {1},
	pages = {35--65},
	pubcountry = {Netherlands},
	title = {Reachability analysis of dynamical systems having piecewise-constant derivatives},
	treatment = {T Theoretical or Mathematical},
	volume = {138},
	year = {1995}}

@article{puri2000dynamical,
	author = {Puri, Anuj},
	journal = {Discrete Event Dynamic Systems},
	pages = {87--113},
	publisher = {Springer},
	title = {Dynamical properties of timed automata},
	volume = {10},
	year = {2000}}

@inproceedings{asarin01perturbed,
	address = {Los Alamitos, CA},
	author = {Asarin, Eugene and Bouajjani, Ahmed},
	booktitle = {Proceedings of the 16th Annual {IEEE} Symposium on Logic in Computer Science ({LICS}-01)},
	month = jun # {~16--19},
	pages = {269--278},
	publisher = {IEEE Computer Society Press},
	title = {Perturbed {Turing} Machines and Hybrid Systems},
	year = {2001}}

@article{KCG94,
	author = {Koiran, Pascal and Cosnard, Michel and Garzon, Max},
	bibdate = {Sat Nov 22 13:15:15 MST 1997},
	classification = {C1230D (Neural nets); C4210 (Formal logic); C4220 (Automata theory)},
	coden = {TCSCDI},
	corpsource = {Lab. de l'Inf. du Parallelisme, Ecole Normale Sup{\'e}rieure de Lyon, France},
	day = {26},
	issn = {0304-3975},
	journal = {Theoretical Computer Science},
	month = sep,
	number = {1-2},
	pages = {113--128},
	pubcountry = {Netherlands},
	title = {Computability with low-dimensional dynamical systems},
	treatment = {T Theoretical or Mathematical},
	volume = {132},
	year = {1994}}

@article{Moo91,
	author = {Moore, Cristopher},
	journal = {Nonlinearity},
	number = 3,
	pages = {199-230},
	title = {Generalized shifts: unpredictability and undecidability in dynamical systems},
	volume = 4,
	year = 1991}

@article{BlancBournezJoC,
	title = {Quantifying the robustness of dynamical systems. Relating time and space to length and precision},
	journal = {Journal of Complexity},
	volume = {96},
	pages = {102055},
	year = {2026},
	issn = {0885-064X},
	doi = {https://doi.org/10.1016/j.jco.2026.102055},
	url = {https://www.sciencedirect.com/science/article/pii/S0885064X26000403},
	author = {Manon Blanc and Olivier Bournez}
}

@article{etessami,
 URL = {https://link.springer.com/article/10.1007/s00182-007-0095-0#citeas},
 author = {Etessami, K. and Lochbihler, A.},
 journal = {International Journal of Game Theory},
 pages = {93--113},
 editor = {Springer},
 title = {The computational complexity of evolutionarily stable strategies},
 volume = {37},
 year = {2007}
}

@InProceedings{CSR21,
author="Blanc, Manon
and Hansen, Kristoffer Arnsfelt",
editor="Santhanam, Rahul
and Musatov, Daniil",
title="Computational Complexity of Multi-player Evolutionarily Stable Strategies",
booktitle="Computer Science -- Theory and Applications",
year="2021",
publisher="Springer International Publishing",
address="Cham",
pages="1--17",
isbn="978-3-030-79416-3"
}

@InProceedings{takeuti,
author="Takeuti, Izumi",
editor="Blanck, Jens
and Brattka, Vasco
and Hertling, Peter",
title="Effective Fixed Point Theorem over a Non-Computably Separable Metric Space",
booktitle="Computability and Complexity in Analysis",
year="2001",
publisher="Springer Berlin Heidelberg",
address="Berlin, Heidelberg",
pages="310--322",
isbn="978-3-540-45335-2"
}

@article{KawCook12,
  author       = {Akitoshi Kawamura and
                  Stephen A. Cook},
  title        = {Complexity Theory for Operators in Analysis},
  journal      = {{ACM} Trans. Comput. Theory},
  volume       = {4},
  number       = {2},
  pages        = {5:1--5:24},
  year         = {2012},
  url          = {https://doi.org/10.1145/2189778.2189780},
  doi          = {10.1145/2189778.2189780},
  bibsource    = {dblp computer science bibliography, https://dblp.org}
}

@phdthesis{pouly,
  TITLE = {{Continuous models of computation: from computability to complexity}},
  AUTHOR = {Pouly, Amaury},
  URL = {https://pastel.hal.science/tel-01223284},
  SCHOOL = {{Ecole Doctorale Polytechnique ; Universidad do Algarve}},
  YEAR = {2015},
  MONTH = Jul,
  TYPE = {Theses},
  HAL_ID = {tel-01223284},
  HAL_VERSION = {v2},
}

@book{BCSS98,
  author = {Lenore Blum and Cucker, Felipe and Shub, Michael and Smale, Steve},
  date-added = {2007-09-03 14:37:44 +0200},
  date-modified = {2019-07-29 15:00:40 +0200},
  publisher = {Springer},
  title = {Complexity and real computation},
  year = {1998},
  timestamp = {Fri, 17 Jul 2020 16:12:44 +0200},
  biburl = {https://dblp.org/rec/books/lib/Blum98.bib},
  bibsource = {dblp computer science bibliography, https://dblp.org},
  isbn = {0387982817},
  url = {https://www.worldcat.org/oclc/37004484},
  _bib2doi_selected = {dblp:/rec/books/lib/Blum98.bib},
  _bib2doi_confirmed = {true},
  _bib2doi_finished = {true},
}

@article{BSS89,
  author = {Blum, Lenore and Shub, Mike and Smale, Steve},
  date-added = {2007-09-03 14:37:44 +0200},
  date-modified = {2019-07-29 15:05:10 +0200},
  file = {BSSComplexite.bib},
  journal = {Bulletin of the American Mathematical Society},
  month = {jul},
  number = {1},
  pages = {1--46},
  title = {On a theory of computation and complexity over the real numbers; {NP} completeness, recursive functions and universal machines},
  volume = {21},
  year = {1989},
  _bib2doi_selected = {dblp:/rec/conf/focs/BlumSS88.bib},
  _bib2doi_finished = {true},
}

@article{Bra95,
  author = {Branicky, M. S.},
  date-added = {2008-01-08 22:02:29 +0100},
  date-modified = {2008-01-08 22:02:29 +0100},
  file = {systemesdynamiques.bib},
  journal = {Theoretical Computer Science},
  local-url = {file://localhost/Users/bournez/BASE-DE-DONNEES/base-devon-pro.dtBase/Files/P-2218-29812798.pdf},
  month = {6~feb},
  number = {1},
  pages = {67--100},
  title = {Universal computation and other capabilities of hybrid and continuous dynamical systems},
  volume = {138},
  year = {1995},
  timestamp = {Wed, 17 Feb 2021 21:56:52 +0100},
  biburl = {https://dblp.org/rec/journals/tcs/Branicky95.bib},
  bibsource = {dblp computer science bibliography, https://dblp.org},
  doi = {10.1016/0304-3975(94)00147-B},
  _bib2doi_selected = {dblp:/rec/journals/tcs/Branicky95.bib},
  _bib2doi_confirmed = {true},
}

@incollection{brattka2008tutorial,
  author = {Brattka, Vasco and Hertling, Peter and Weihrauch, Klaus},
  booktitle = {New computational paradigms},
  date-added = {2018-01-27 12:11:02 +0000},
  date-modified = {2018-01-27 12:11:02 +0000},
  pages = {425--491},
  publisher = {Springer},
  title = {A tutorial on computable analysis},
  year = {2008},
  _bib2doi_finished = {true},
}

@inproceedings{chen2018neural,
  author = {Chen, Tian Qi and Rubanova, Yulia and Bettencourt, Jesse and Duvenaud, David K},
  booktitle = {Advances in Neural Information Processing Systems},
  date-added = {2019-05-05 22:23:33 +0200},
  date-modified = {2019-05-05 22:23:33 +0200},
  pages = {6571--6583},
  title = {Neural ordinary differential equations},
  year = {2018},
  timestamp = {Mon, 16 May 2022 15:41:51 +0200},
  biburl = {https://dblp.org/rec/conf/nips/ChenRBD18.bib},
  bibsource = {dblp computer science bibliography, https://dblp.org},
  url = {https://proceedings.neurips.cc/paper/2018/hash/69386f6bb1dfed68692a24c8686939b9-Abstract.html},
  _bib2doi_selected = {dblp:/rec/conf/nips/ChenRBD18.bib},
  _bib2doi_confirmed = {true},
}

@inproceedings{CMSB17vantardise,
  author = {Fran{\c{c}}ois Fages and Guillaume Le Guludec and Olivier Bournez and Amaury Pouly},
  booktitle = {Computational Methods in Systems Biology - 15th International Conference, {CMSB} 2017, Darmstadt, Germany, September 27-29, 2017, Proceedings},
  date-added = {2017-06-29 08:56:26 +0000},
  date-modified = {2022-09-24 09:55:56 +0200},
  note = {{CMSB'2017 Best Paper Award}},
  title = {Strong Turing Completeness of Continuous Chemical Reaction Networks and Compilation of Mixed Analog-Digital Programs},
  year = {2017},
  timestamp = {Tue, 14 May 2019 10:00:49 +0200},
  biburl = {https://dblp.org/rec/conf/cmsb/FagesGBP17.bib},
  bibsource = {dblp computer science bibliography, https://dblp.org},
  doi = {10.1007/978-3-319-67471-1_7},
  publisher = {Springer},
  volume = {10545},
  pages = {108--127},
  url = {https://doi.org/10.1007/978-3-319-67471-1\_7},
  editor = {J{\'{e}}r{\^{o}}me Feret and Heinz Koeppl},
  series = {Lecture Notes in Computer Science},
  _bib2doi_selected = {dblp:/rec/conf/cmsb/FagesGBP17.bib},
  _bib2doi_confirmed = {true},
  _bib2doi_finished = {true},
}

@phdthesis{TheseRiccardo,
  author = {Riccardo Gozzi},
  date-added = {2022-02-11 16:16:51 +0100},
  date-modified = {2022-02-11 16:31:46 +0100},
  school = {Instituto Superior T{\'e}cnico, Lisbon, Portugal and University of Algarve, Faro, Portugal},
  title = {Analog Characterization of Complexity Classes},
  year = {2022},
  _bib2doi_finished = {true},
}

@article{GraCos03,
  abstract = {In this paper we show that Shannon's General Purpose Analog Computer (GPAC) is equivalent to a particular class of recursive functions over the reals with the flavour of Kleene's classical recursive function theory. We first consider the GPAC and several of its extensions to show that all these models have drawbacks and we introduce an alternative continuous-time model of computation that solves these problems. We also show that this new model preserves all the significant relations involving the previous models (namely, the equivalence with the differentially algebraic functions). We then continue with the topic of recursive functions over the reals, and we show full connections between functions generated by the model introduced so far and a particular class of recursive functions over the reals.},
  author = {Gra{\c c}a, Daniel S. and Costa, Jos{\'e} F{\'e}lix},
  date-added = {2008-01-24 09:24:30 +0100},
  date-modified = {2008-01-24 09:24:31 +0100},
  journal = {jcomp},
  local-url = {file://localhost/Users/bournez/BASE-DE-DONNEES/base-devon-pro.dtBase/Files/GraCos03.pdf},
  number = {5},
  pages = {644--664},
  title = {Analog computers and recursive functions over the reals},
  volume = {19},
  year = {2003},
  timestamp = {Tue, 16 Feb 2021 08:55:31 +0100},
  biburl = {https://dblp.org/rec/journals/jc/GracaC03.bib},
  bibsource = {dblp computer science bibliography, https://dblp.org},
  doi = {10.1016/S0885-064X(03)00034-7},
  _bib2doi_selected = {dblp:/rec/journals/jc/GracaC03.bib},
  _bib2doi_confirmed = {true},
}

@inbook{GracaZhongHandbook,
  author = {Gra{\c c}a, Daniel S. and Zhong, Ning},
  chapter = {Computability of Differential Equations},
  date-added = {2019-07-17 13:53:14 +0200},
  date-modified = {2019-07-29 14:58:40 +0200},
  journal = {To appear.},
  publisher = {Springer.},
  title = {Handbook of Computability and Complexity in Analysis},
  topics = {Analog,Dynamical,GPAC},
  year = {2018},
  _bib2doi_finished = {true},
}

@book{HSD03,
  author = {Hirsch, Morris W. and Smale, Stephen and Devaney, Robert},
  date-added = {2007-09-03 14:37:44 +0200},
  date-modified = {2007-09-03 14:37:44 +0200},
  publisher = {Elsevier Academic Press},
  title = {Differential Equations, Dynamical Systems, and an Introduction to Chaos},
  year = {2003},
  _bib2doi_finished = {true},
}

@conference{kawamura2009lipschitz,
  author = {Kawamura, Akitoshi},
  booktitle = {2009 24th Annual IEEE Conference on Computational Complexity},
  date-added = {2011-04-13 21:09:23 +0200},
  date-modified = {2019-07-29 14:51:37 +0200},
  organization = {IEEE},
  pages = {149--160},
  title = {{Lipschitz continuous ordinary differential equations are polynomial-space complete}},
  year = {2009},
  timestamp = {Fri, 24 Mar 2023 00:04:21 +0100},
  biburl = {https://dblp.org/rec/conf/coco/Kawamura09.bib},
  bibsource = {dblp computer science bibliography, https://dblp.org},
  doi = {10.1109/CCC.2009.34},
  _bib2doi_selected = {dblp:/rec/conf/coco/Kawamura09.bib},
  _bib2doi_confirmed = {true},
}

@article{kawamura2014computational,
  author = {Kawamura, Akitoshi and Ota, Hiroyuki and R{\"o}snick, Carsten and Ziegler, Martin},
  date-added = {2024-02-11 15:41:08 +0100},
  date-modified = {2024-02-11 15:41:08 +0100},
  journal = {Logical Methods in Computer Science},
  publisher = {Episciences. org},
  title = {Computational complexity of smooth differential equations},
  volume = {10},
  year = {2014},
  timestamp = {Thu, 25 Jun 2020 21:29:08 +0200},
  biburl = {https://dblp.org/rec/journals/corr/KawamuraORZ13.bib},
  bibsource = {dblp computer science bibliography, https://dblp.org},
  doi = {10.2168/LMCS-10(1:6)2014},
  _bib2doi_selected = {dblp:/rec/journals/corr/KawamuraORZ13.bib},
  _bib2doi_confirmed = {true},
}

@article{kidger2022neural,
  author = {Patrick Kidger},
  date-added = {2022-02-16 14:52:39 +0100},
  date-modified = {2022-02-16 14:52:39 +0100},
  journal = {CoRR},
  title = {On Neural Differential Equations},
  year = {2022},
  timestamp = {Wed, 09 Feb 2022 15:43:35 +0100},
  biburl = {https://dblp.org/rec/journals/corr/abs-2202-02435.bib},
  bibsource = {dblp computer science bibliography, https://dblp.org},
  doi = {10.48550/arXiv.2202.02435},
  volume = {abs/2202.02435},
  url = {https://arxiv.org/abs/2202.02435},
  eprinttype = {arXiv},
  eprint = {2202.02435},
  _bib2doi_selected = {dblp:/rec/journals/corr/abs-2202-02435.bib},
  _bib2doi_confirmed = {true},
  _bib2doi_finished = {true},
}

@book{Ko91,
  address = {Boston},
  author = {Ker{-}I Ko},
  date-added = {2007-09-03 14:37:44 +0200},
  date-modified = {2008-01-08 22:04:51 +0100},
  file = {analog.bib},
  publisher = {Birkh{\"{a}}user},
  series = {Progress in theoretical computer science},
  title = {Complexity theory of real functions},
  year = {1991},
  timestamp = {Wed, 17 Jul 2019 17:20:52 +0200},
  biburl = {https://dblp.org/rec/books/daglib/0067010.bib},
  bibsource = {dblp computer science bibliography, https://dblp.org},
  isbn = {978-0-8176-3586-2},
  volume = {3},
  _bib2doi_selected = {dblp:/rec/books/daglib/0067010.bib},
  _bib2doi_confirmed = {true},
  _bib2doi_finished = {true},
}

@article{Sha41,
  author = {Shannon, Claude E.},
  date-added = {2007-09-03 14:37:44 +0200},
  date-modified = {2019-07-29 14:47:19 +0200},
  file = {analog.bib},
  journal = {Journal of Mathematics and Physics MIT},
  pages = {337-354},
  title = {Mathematical Theory of the Differential Analyser},
  volume = {20},
  year = {1941},
  _bib2doi_finished = {true},
}

@book{Sip97,
  author = {Sipser, Michael},
  date-added = {2007-09-03 14:37:44 +0200},
  date-modified = {2019-07-29 14:46:54 +0200},
  publisher = {PWS Publishing Company},
  title = {Introduction to the Theory of Computation},
  year = {1997},
  timestamp = {Thu, 21 Apr 2011 19:59:45 +0200},
  biburl = {https://dblp.org/rec/books/daglib/0086373.bib},
  bibsource = {dblp computer science bibliography, https://dblp.org},
  isbn = {978-0-534-94728-6},
  _bib2doi_selected = {dblp:/rec/books/daglib/0086373.bib},
  _bib2doi_confirmed = {true},
}

@article{Tur36,
  author = {Alan M. Turing},
  date-added = {2007-09-03 14:37:44 +0200},
  date-modified = {2009-11-20 22:13:12 +0100},
  file = {analog.bib},
  journal = {Proceedings of the London Mathematical Society},
  note = {Reprinted in Martin Davis. The Undecidable: Basic Papers on Undecidable Propositions, Unsolvable Problems and Computable Functions. Raven Press, 1965.},
  number = {1},
  pages = {230--265},
  title = {On computable numbers, with an application to the Entscheidungsproblem},
  volume = {s2-42},
  year = {1937},
  timestamp = {Tue, 16 Aug 2022 23:09:52 +0200},
  biburl = {https://dblp.org/rec/journals/x/Turing37.bib},
  bibsource = {dblp computer science bibliography, https://dblp.org},
  doi = {10.1112/plms/s2-42.1.230},
  url = {https://doi.org/10.1112/plms/s2-42.1.230},
  urn = {urn:nbn:de:hbz:6:1-85171},
  _bib2doi_selected = {dblp:/rec/journals/x/Turing37.bib},
  _bib2doi_confirmed = {true},
  _bib2doi_finished = {true},
}

@book{LivreAnalogcomputing,
  author = {Ulmann, Bernd},
  date-added = {2016-02-11 08:16:30 +0000},
  date-modified = {2016-02-11 08:16:42 +0000},
  publisher = {Walter de Gruyter},
  title = {Analog computing},
  year = {2013},
  _bib2doi_finished = {true},
}

@book{ulmann2020analog,
  author = {Ulmann, Bernd},
  date-added = {2022-02-03 12:41:38 +0100},
  date-modified = {2022-02-03 12:41:38 +0100},
  publisher = {De Gruyter Oldenbourg},
  title = {Analog and hybrid computer programming},
  year = {2020},
  _bib2doi_finished = {true},
}

@book{Wei00,
  author = {Klaus Weihrauch},
  date-added = {2007-09-03 14:37:44 +0200},
  date-modified = {2019-07-29 14:46:02 +0200},
  publisher = {Springer},
  title = {Computable Analysis - An Introduction},
  year = {2000},
  timestamp = {Tue, 16 May 2017 14:24:38 +0200},
  biburl = {https://dblp.org/rec/series/txtcs/Weihrauch00.bib},
  bibsource = {dblp computer science bibliography, https://dblp.org},
  doi = {10.1007/978-3-642-56999-9},
  url = {https://doi.org/10.1007/978-3-642-56999-9},
  series = {Texts in Theoretical Computer Science. An {EATCS} Series},
  isbn = {978-3-540-66817-6},
  _bib2doi_selected = {dblp:/rec/series/txtcs/Weihrauch00.bib},
  _bib2doi_confirmed = {true},
  _bib2doi_finished = {true},
}

@inproceedings{BlancBournezMCU22vantardise,
  author = {Manon Blanc and Olivier Bournez},
  bibsource = {dblp computer science bibliography, https://dblp.org},
  biburl = {https://dblp.org/rec/conf/mcu/BlancB22.bib},
  booktitle = {Machines, Computations, and Universality - 9th International Conference, {MCU} 2022, Debrecen, Hungary, August 31 - September 2, 2022, Proceedings},
  date-added = {2022-09-16 13:53:36 +0200},
  date-modified = {2024-02-13 21:49:29 +0100},
  doi = {10.1007/978-3-031-13502-6_4},
  editor = {J{\'{e}}r{\^{o}}me Durand{-}Lose and Gy{\"{o}}rgy Vaszil},
  pages = {58--74},
  publisher = {Springer},
  series = {Lecture Notes in Computer Science},
  timestamp = {Tue, 09 Aug 2022 17:12:09 +0200},
  title = {A Characterization of Polynomial Time Computable Functions from the Integers to the Reals Using Discrete Ordinary Differential Equations},
  url = {https://doi.org/10.1007/978-3-031-13502-6\_4},
  volume = {13419},
  year = {2022},
  bdsk-url-1 = {https://doi.org/10.1007/978-3-031-13502-6%5C_4},
  _bib2doi_old_doi = {10.1007/978-3-031-13502-6\_4},
  _bib2doi_selected = {dblp:/rec/conf/mcu/BlancB22.bib},
  _bib2doi_confirmed = {true},
  _bib2doi_finished = {true},
}

@inproceedings{BlancBournezMFCS2023vantardise,
  author = {Manon Blanc and Olivier Bournez},
  booktitle = {48th International Symposium on Mathematical Foundations of Computer Science, {MFCS} 2023, August 28 to September 1, 2023, Bordeaux, France},
  date-added = {2024-02-13 21:51:15 +0100},
  date-modified = {2024-02-13 21:58:15 +0100},
  organization = {Schloss Dagstuhl-Leibniz-Zentrum f{\"u}r Informatik},
  title = {A Characterisation of Functions Computable in Polynomial Time and Space over the Reals with Discrete Ordinary Differential Equations: Simulation of Turing Machines with Analytic Discrete ODEs},
  volume = {272},
  year = {2023},
  timestamp = {Mon, 28 Aug 2023 17:23:09 +0200},
  biburl = {https://dblp.org/rec/conf/mfcs/BlancB23.bib},
  bibsource = {dblp computer science bibliography, https://dblp.org},
  doi = {10.4230/LIPIcs.MFCS.2023.21},
  publisher = {Schloss Dagstuhl - Leibniz-Zentrum f{\"{u}}r Informatik},
  pages = {21:1--21:15},
  url = {https://doi.org/10.4230/LIPIcs.MFCS.2023.21},
  editor = {J{\'{e}}r{\^{o}}me Leroux and Sylvain Lombardy and David Peleg},
  series = {LIPIcs},
  _bib2doi_selected = {dblp:/rec/conf/mfcs/BlancB23.bib},
  _bib2doi_confirmed = {true},
  _bib2doi_finished = {true},
}

@inproceedings{MFCS2019,
  author = {Olivier Bournez and Arnaud Durand},
  booktitle = {44th Int Symposium on Mathematical Foundations of Computer Science, {MFCS}},
  date-added = {2019-06-14 20:18:06 +0200},
  date-modified = {2019-06-14 20:18:44 +0200},
  editor = {Rossmanith, Peter and Heggernes, Pinar and Katoen, Joost{-}Pieter},
  pages = {23:1--23:14},
  publisher = {Schloss Dagstuhl - Leibniz-Zentrum f{\"{u}}r Informatik},
  series = {LIPIcs},
  title = {Recursion schemes, discrete differential equations and characterization of polynomial time computation},
  volume = {138},
  year = {2019},
  timestamp = {Thu, 05 Sep 2019 19:42:15 +0200},
  biburl = {https://dblp.org/rec/conf/mfcs/BournezD19.bib},
  bibsource = {dblp computer science bibliography, https://dblp.org},
  doi = {10.4230/LIPIcs.MFCS.2019.23},
  _bib2doi_selected = {dblp:/rec/conf/mfcs/BournezD19.bib},
  _bib2doi_confirmed = {true},
}

@article{MFCSJournal,
  author = {Bournez, Olivier and Durand, Arnaud},
  date-added = {2023-07-18 16:15:15 +0200},
  date-modified = {2023-07-18 16:19:55 +0200},
  journal = {Computational Complexity},
  number = {2},
  pages = {7},
  publisher = {Springer},
  title = {A Characterization of Functions over the Integers Computable in Polynomial Time Using Discrete Ordinary Differential Equations},
  volume = {32},
  year = {2023},
  timestamp = {Fri, 21 Jul 2023 17:55:54 +0200},
  biburl = {https://dblp.org/rec/journals/cc/BournezD23.bib},
  bibsource = {dblp computer science bibliography, https://dblp.org},
  doi = {10.1007/s00037-023-00240-1},
  _bib2doi_selected = {dblp:/rec/journals/cc/BournezD23.bib},
  _bib2doi_confirmed = {true},
}

@article{BGDPRiccardo2022,
  author = {Bournez, Olivier and Gozzi, Riccardo and Gra{\c{c}}a, Daniel S and Pouly, Amaury},
  date-modified = {2023-11-30 08:30:11 +0100},
  journal = {Journal of Complexity},
  month = {august},
  pages = {101755},
  publisher = {Elsevier},
  title = {A continuous characterization of {PSPACE} using polynomial ordinary differential equations},
  url = {https://www.sciencedirect.com/science/article/pii/S0885064X23000249?dgcid)=author},
  volume = {77},
  year = {2023},
  bdsk-url-1 = {https://www.sciencedirect.com/science/article/pii/S0885064X23000249?dgcid)=author},
  timestamp = {Tue, 12 Sep 2023 07:58:30 +0200},
  biburl = {https://dblp.org/rec/journals/jc/BournezGGP23.bib},
  bibsource = {dblp computer science bibliography, https://dblp.org},
  doi = {10.1016/j.jco.2023.101755},
  _bib2doi_selected = {dblp:/rec/journals/jc/BournezGGP23.bib},
  _bib2doi_confirmed = {true},
}

@inproceedings{ICALP2016vantardise,
  author = {Olivier Bournez and Daniel Silva Gra{\c{c}}a and Amaury Pouly},
  booktitle = {43rd International Colloquium on Automata, Languages, and Programming, {ICALP} 2016, July 11-15, 2016, Rome, Italy},
  date-added = {2024-02-13 22:02:09 +0100},
  date-modified = {2024-02-13 22:03:36 +0100},
  pages = {109:1--109:15},
  pdf = {http://www.lix.polytechnique.fr/~bournez/uploads/Main/icalp2016.pdf},
  publisher = {Schloss Dagstuhl - Leibniz-Zentrum f{\"{u}}r Informatik},
  series = {LIPIcs},
  title = {Polynomial Time Corresponds to Solutions of Polynomial Ordinary Differential Equations of Polynomial Length: The General Purpose Analog Computer and Computable Analysis Are Two Efficiently Equivalent Models of Computations},
  topics = {Analog,Dynamical,GPAC},
  volume = {55},
  year = {2016},
  bdsk-url-1 = {http://drops.dagstuhl.de/opus/frontdoor.php?source_opus=6244},
  timestamp = {Wed, 29 Jul 2020 09:29:28 +0200},
  biburl = {https://dblp.org/rec/conf/icalp/BournezGP16.bib},
  bibsource = {dblp computer science bibliography, https://dblp.org},
  doi = {10.4230/LIPIcs.ICALP.2016.109},
  url = {https://doi.org/10.4230/LIPIcs.ICALP.2016.109},
  editor = {Ioannis Chatzigiannakis and Michael Mitzenmacher and Yuval Rabani and Davide Sangiorgi},
  _bib2doi_selected = {dblp:/rec/conf/icalp/BournezGP16.bib},
  _bib2doi_confirmed = {true},
  _bib2doi_finished = {true},
}

@inproceedings{ICALP2017,
  author = {Olivier Bournez and Amaury Pouly},
  booktitle = {44th International Colloquium on Automata, Languages, and Programming, {ICALP} 2017, July 10-14, 2017, Warsaw, Poland},
  date-added = {2017-04-14 08:10:15 +0000},
  date-modified = {2017-04-14 08:11:12 +0000},
  title = {A Universal Ordinary Differential Equation},
  topics = {Analog,Dynamical,GPAC},
  year = {2017},
  timestamp = {Tue, 11 Feb 2020 15:52:14 +0100},
  biburl = {https://dblp.org/rec/conf/icalp/BournezP17.bib},
  bibsource = {dblp computer science bibliography, https://dblp.org},
  doi = {10.4230/LIPIcs.ICALP.2017.116},
  publisher = {Schloss Dagstuhl - Leibniz-Zentrum f{\"{u}}r Informatik},
  volume = {80},
  pages = {116:1--116:14},
  url = {https://doi.org/10.4230/LIPIcs.ICALP.2017.116},
  editor = {Ioannis Chatzigiannakis and Piotr Indyk and Fabian Kuhn and Anca Muscholl},
  series = {LIPIcs},
  _bib2doi_selected = {dblp:/rec/conf/icalp/BournezP17.bib},
  _bib2doi_confirmed = {true},
  _bib2doi_finished = {true},
}

@incollection{bournez2021survey,
  author = {Olivier Bournez and Amaury Pouly},
  booktitle = {Handbook of Computability and Complexity in Analysis},
  date-added = {2022-02-11 21:50:02 +0100},
  date-modified = {2022-02-11 21:50:02 +0100},
  pages = {173--226},
  publisher = {Springer},
  title = {A survey on analog models of computation},
  year = {2021},
  _bib2doi_finished = {true},
}

@article{Schaefer24,
  author       = {Marcus Schaefer and
                  Jean Cardinal and
                  Tillmann Miltzow},
  title        = {The Existential Theory of the Reals as a Complexity Class: {A} Compendium},
  journal      = {CoRR},
  volume       = {abs/2407.18006},
  year         = {2024},
  url          = {https://doi.org/10.48550/arXiv.2407.18006},
  doi          = {10.48550/ARXIV.2407.18006},
  eprinttype    = {arXiv},
  eprint       = {2407.18006},
  bibsource    = {dblp computer science bibliography, https://dblp.org}
}

@article{collins05,
title = {Continuity and computability of reachable sets},
journal = {Theoretical Computer Science},
volume = {341},
number = {1},
pages = {162-195},
year = {2005},
issn = {0304-3975},
doi = {https://doi.org/10.1016/j.tcs.2005.05.001},
url = {https://www.sciencedirect.com/science/article/pii/S0304397505003038},
author = {Pieter Collins}
}

@book{conley1978,
  title={Isolated Invariant Sets and the Morse Index},
  author={Conley, C.C. and Conference Board of the Mathematical Sciences},
  isbn={9781470423988},
  lccn={78015772},
  series={Regional conference series in mathematics},
  url={https://books.google.fr/books?id=7u9KAQAACAAJ},
  year={1978},
  publisher={Conference Board of the Mathematical Sciences}
}

@InProceedings{CSL24,
  author =	{Blanc, Manon and Bournez, Olivier},
  title =	{{Quantifiying the Robustness of Dynamical Systems. Relating Time and Space to Length and Precision}},
  booktitle =	{32nd EACSL Annual Conference on Computer Science Logic (CSL 2024)},
  pages =	{17:1--17:20},
  series =	{Leibniz International Proceedings in Informatics (LIPIcs)},
  ISBN =	{978-3-95977-310-2},
  ISSN =	{1868-8969},
  year =	{2024},
  volume =	{288},
  editor =	{Murano, Aniello and Silva, Alexandra},
  publisher =	{Schloss Dagstuhl -- Leibniz-Zentrum f{\"u}r Informatik},
  address =	{Dagstuhl, Germany},
  URL =		{https://drops.dagstuhl.de/entities/document/10.4230/LIPIcs.CSL.2024.17},
  URN =		{urn:nbn:de:0030-drops-196609},
  doi =		{10.4230/LIPIcs.CSL.2024.17}
}

@InProceedings{ICALP24,
  author =	{Blanc, Manon and Bournez, Olivier},
  title =	{{The Complexity of Computing in Continuous Time: Space Complexity Is Precision}},
  booktitle =	{51st International Colloquium on Automata, Languages, and Programming (ICALP 2024)},
  pages =	{129:1--129:22},
  series =	{Leibniz International Proceedings in Informatics (LIPIcs)},
  ISBN =	{978-3-95977-322-5},
  ISSN =	{1868-8969},
  year =	{2024},
  volume =	{297},
  editor =	{Bringmann, Karl and Grohe, Martin and Puppis, Gabriele and Svensson, Ola},
  publisher =	{Schloss Dagstuhl -- Leibniz-Zentrum f{\"u}r Informatik},
  address =	{Dagstuhl, Germany},
  URL =		{https://drops.dagstuhl.de/entities/document/10.4230/LIPIcs.ICALP.2024.129},
  URN =		{urn:nbn:de:0030-drops-202722},
  doi =		{10.4230/LIPIcs.ICALP.2024.129}
}
	\appendix
	
	\section{Little course of computable analysis}\label{secCA}
	\subsection{Notions of computable analysis} 
	
	In classical models of computation, we usually consider an input of finite size, e.g. integers or graphs, and work on some representation on a finite alphabet. Such objects are said to be \emph{computable} when there exists some Turing machine producing their representations.
	
	The main idea behind computable analysis is to be able to talk about objects which usually do not have a finite representation, such as real numbers, functions over the reals, compact sets, etc. A reason why we cannot use classical models is that the set of words on a finite alphabet is countable, while $\R$ is not. 
	
	However, they can be represented by some infinite words over a finite alphabet and the idea is to fix such representations for these various objects, called \emph{names}, with suitable computable properties. In particular, in all the following proposed representations, it was proved that an object is computable iff it has some computable representation.	
We denote by $\Sigma$ the machine alphabet, $\Sigma^*$ denotes the finite words on this alphabet (they can be arbitrarly long) and $\Sigma^\omega$ the infinite words on $\Sigma$.
	
For discrete-time models of computations over the reals, the most famous models are computable analysis, based on the Turing machine model (see \cite{Tur36}, \cite{Wei00}), and algebraic models such as the Blum Shub Smale (BSS) model of computation \cite{BSS89}, \cite{BCSS98}. 
The class $\exists \R$ corresponds to the (constant-free, equivalently uniform) non-deterministic time of the BSS model of computation. Numerous decision problems have been proved recently in this class \cite{Schaefer24}.
Both models were tailored for different applications, and it is well known that we cannot unify existing models with the equivalent of a Church-Turing thesis. 
For example, computable functions in a computable analysis model must be continuous, whereas the BSS model aims to consider functions and problems over polynomials that are not. The former model has not been introduced to be associated with actual physical machines.	
Among models of computation over the reals, we can highlight continuous-time models. This includes models of old, first-ever-built computers, such as the Differential Analysers \cite{ulmann2020analog}. A famous mathematical model of such machines is the General Purpose Analog Computer model of Claude Shannon \cite{Sha41}. It covers many historical machines and today's analogue devices \cite{LivreAnalogcomputing} too. It also includes various recent approaches and models from deep learning, such as Neural ODEs \cite{chen2018neural,kidger2022neural} with many variants. 

In the context of continuous-time, the situation is clearer than with discrete-time models, as there is a unifying way to describe these models, provided by Ordinary Differential Equations (ODEs). Each model corresponds to a particular class of  ODEs.  For example, the GPAC corresponds to polynomial ODEs \cite{GraCos03}, and Neural ODEs are made by selecting
the best solution among a parameterised class of ODEs: see, e.g. \cite{kidger2022neural}. 

When we say that a function $f: S_{1}  \times \cdots \times S_{d} \to \R^{d'}$, $S_i \subseteq\R$ for all $i\in\textlbrackdbl 1, d\textrbrackdbl $, is (respectively: polynomial-time) computable, it will always be in the sense of computable analysis. We recall here the basic concepts and definitions, mostly following the book \cite{Ko91} on complexity theory in computable analysis. Alternative presentations include \cite{brattka2008tutorial,Wei00}. 
	
	\begin{remark}
	The monograph \cite{Ko91} formulates explicitly most of the statements only for functions over the reals. 
	Since we talk about functions in $\R^\N$, we need to mix complexity issues dealing with integer and real arguments. Hence, we extend some of the statements to these more general settings (in the spirit of \cite{Ko91,kawamura2009lipschitz,kawamura2014computational,kawamura2015computational,kawamura2018parameterized}).
\end{remark}

We present here a formalisation equivalent to \cite{takeuti} but with our own words. This will be used to fix the framework and the notations.

\begin{definition}[Computable String Function $f$ \cite{Wei00}]
	The initial tape configuration for input $(y_1, \cdots, y_k) \in Y_1\times\dots\times Y_k$ is as follows: for each input tape $i$, the (finite or infinite) sequence $y_i\in Y_i$ is placed on the tape immediately to the right of the head, all other tape cells contain the symbol $\mathtt{B}$. For all $y_0\in Y_0, y_1\in Y_1, \cdots, y_k\in Y_k$ we define:
	\begin{enumerate}
		\item \underline{\textbf{Case $Y_0 = \Sigma^*$}:}
		$f(y_1,\cdots, y_k) := y_0 \in \Sigma^*$ iff the associated Type-2 Turing machine $M$ halts on input $(y_1,\cdots,y_k)$ with $y_0$ written on the output tape.
		
		\item \underline{\textbf{Case $Y_0 = \Sigma^\omega$}:}
		$f(y_1,\cdots, y_k) := y_0 \in \Sigma^\omega$ iff $M$ computes forever on input $(y_1,\cdots,y_k)$ and writes $y_0$ on the output tape.
	\end{enumerate} 	 	
\end{definition} 

We can now define a proper notion of computable enumerability in that context:

\begin{definition}[C.e. sets \cite{Wei00}]
	A closed set $A\in\R^n$ is \emph{c.e.} if there exists a Turing machine that enumerates all the open rational cubes intersecting $A$.
\end{definition}

	We can also consider oracle Turing machines to represent functions in $\R^\R$, as we will see in Definition \ref{MBCompFun}.

We approximate real numbers with dyadics. A dyadic number $d$ is a rational number with a finite binary expansion. That is to say $d=m / 2^{n}$ for some integers $m \in \Z$, $n\in \N$, $n \geq 0$. Let $\dyadic$ be the set of all dyadic rational numbers.

We denote by $\dyadic_{n}$ the set of all dyadic rationals $d$ with a representation $s$ of precision $\operatorname{prec}(s)=n$: that is, $\dyadic_{n}=\left\{m ~\cdot~ 2^{-n} \mid m \in \Z\right\}$.
	
	\begin{remark}
		$\dyadic$ is dense in $\R$.
	\end{remark}
	
		Thus, we can approximate any real number $x$ with a sequence of dyadics converging to $x$. 
		
		$C F_{x}$ is a set of Cauchy sequences converging to $x$:
		\begin{definition}[Cauchy representation of a real number]
			Define the Cauchy representation $\rho_C : \Sigma^\omega \rightarrow \R$ by $\rho_C(p) = x$ iff there are words $w_0, w_1, \cdots \in \text{Dom}(\nu_\Q)$ such that $ p = \iota(w_0)\iota(w_1) \cdots$, $\left| \bar{w_i} - \bar{w_k} \right| \leq 2^{-i}$ for $i<k$ and $x = \lim_{i \rightarrow \infty} \bar{w_i}$.
		\end{definition}

		It leads us to the definition of computable real numbers:

\begin{definition}[{\cite{Ko91}}]  \label{def:cinq} 
	For each real number $x$, a function $\phi: \N \rightarrow \dyadic$ is said to binary converge to $x$ if  for all $n \in \N, \operatorname{prec}(\phi(n))=n$ and $|\phi(n)-x| \leq 2^{-n}$. 
\end{definition}

Thus, we deduce:

\begin{definition}[Computable real number \cite{Ko91}]
	A real number $x$ is computable iff $CF_x$ contains a computable function: there exists a TM such that, on input $n\in\N$, computes the $n$ first digits of $x$.
\end{definition}

There exist many more non-computable than computable real numbers. One intuition is to notice that, according to the previous definition, there is an injection between the set of Turing machines and computable reals: the set of computable real numbers is countable, whereas $\R$ is not. However, the set of computable reals being dense in $\R$, any non-computable real can be approximated by a sequence of computable reals. 
A \emph{name} of $x \in \R^{d}$ is a sequence $\left(I_{n}\right)_{n\in \N}$ of nested open rational balls  with $I_{n+1} \subseteq I_n$ for all $n \in \N$ and $\{x\}=\bigcap_{n \in \N} I_n$  
Such a name can be encoded as an infinite sequence of symbols.

In \cite{BlancBournezMFCS2023vantardise} and \cite{ICALP24}, the authors are interested in having characterisations of functions computable in $\FPtime$ and $\FPspace$. So, we have to formally define what a computable function is in the framework of computable analysis:

\begin{definition}[\cite{Ko91}] \label{MBCompFun}
	A real function $f: \R \rightarrow \R$ is computable if there is a function-oracle {TM} $M$ such that for each $x \in \R$ and each $\phi \in C F_{x}$, the function $\psi$ computed by $M$ with oracle $\phi$  (i.e. $\psi(n)=M^{\phi}(n))$ is in $C F_{f(x)}$. 
	We say the function $f$ is computable on the interval $[a, b]$ if the above condition holds for all $x \in[a, b]$.
\end{definition}

We call a real function $f: \subseteq \mathbb{R} \rightarrow \mathbb{R}$ computable, iff some Type-2 Turing machine maps any name of any $x \in \dom(f)$ to a name of $f(x)$. For real functions $\tu f: \subseteq \mathbb{R}^k \rightarrow \mathbb{R}$ we consider machines reading $k$ names in parallel.  A computable function is necessarily continuous: see \cite{Wei00} for all details.

	The intuition is that an oracle Turing machine $M$ computes a function $f$ over the real, if we have, for any real $x$ in the domain of $f$, for any
	$\phi \in C F_{x}$, $M$ with oracle $\phi$, the following: given the output precision $2^{-n}$, given in the form of integer $n$, $M$ eventually output $d \in \dyadic$ with $\|d-f(x)\| \leq$ $2^{-n}$. During this computation, $M$  asks necessarily finitely many questions to its oracle $\phi$. From the considered oracles, any such question can be seen as follows: for some queried precision $2^{-m}$,  $M$ asks the oracle $\phi$  to provide some $\phi(m)$, such that $\|\phi(m)-x\| \leq 2^{-m}$.
	
	We notice that, in computable analysis, we only deal with continuous functions:
	
	\begin{theorem}[\cite{Wei00}]
		Every computable real function is continuous.
	\end{theorem}
	
	\begin{remark}\label{NBeqatest}
		From the model of computable analysis, given the name of $\tu f$, $\vx, \vy\in\Q$, it is impossible in general to tell effectively if  $\tu f(\tu \vx)=\vy$.
	\end{remark}
	
	For closed sets, the notion of computability can be interpreted as the possibility of being plotted with an arbitrarily chosen precision: $\vz/2^{n}$ corresponds to a pixel at precision $2^{-n}$, $1$ is black (the pixel is plotted black), $0$ is white (the pixel is plotted white). 
	
	\begin{theorem}[\cite{brattka2008tutorial}] \label{thdessin}
		For a closed set $A \subseteq \mathbb{R}^k$,  A is computable iff it can be plotted: there exists a computable function $ f: \N \times \Z^k \rightarrow \{0,1\}$ and such that for all $n \in \mathbb{N}$ and $\vz \in \mathbb{Z}^k$
		$$
		f(n, \vz)= \begin{cases}1 & \text { if } B\left(\frac{\vz}{2^n},2^{-n}\right) \cap A \neq \emptyset,
			\\ 0 & \text { if } B\left(\frac{\vz}{2^n},2\cdot2^{-n}\right) \cap A = \emptyset, 
			\\ 0 \text { or } 1 & \text { otherwise. }\end{cases}
		$$
	\end{theorem}

Computability and complexity classes over the reals depend on the model of computation we consider, and also the representations of the reals: some complexity results hold for some representation (of the reals) but not for another (see \cite{Wei00}). Here, we consider only equivalent, with respect to the following definition, representations.

\begin{definition}[Translation and equivalence, from \cite{Wei00}]\label{defRepEquiv}
	For arbitrary functions $\gamma: Y \rightarrow S $, $\gamma^{'}: Y^{'} \rightarrow S^{'}$ with $Y, Y^{'} \in \{\Sigma^*, \Sigma^\omega\}$ we have:
	$f:  S^Y \rightarrow {S^{'}}^{Y^{'}}$, $f$ computable, \emph{translates} $\gamma$ to $\gamma^{'}$ iff, for all $y\in \dom(\gamma)$, $\gamma(y) = \gamma^{'}(f(y))$; and $\gamma$ is \emph{equivalent} to $\gamma^{'}$ iff some computable function translates $\gamma$ to $\gamma^{'}$ and some computable function translates $\gamma^{'}$ to $\gamma$. 
\end{definition}

	\subsection{On computable analysis: Complexity}
	
	Assume that $M$ is an oracle machine which computes $f$ on the domain $\dom(f)$. 

For simplicity, we may assume in our discussion that $\dom(f)$ is some closed interval or $\R$. For any oracle $\phi \in C F_{x}$, with $x \in \dom(f)$, let $T_{M}(\phi,~ n)$ be the number of steps for $M$ to halt on input $n$ with oracle $\phi$, and $T_{M}^{\prime}(x,~ n)=\max \left\{T_{M}(\phi,~ n) \mid \phi \in C F_{x}\right\}$. 

We define the time complexity of $f$ as follows:

\begin{definition}[{From \cite{Ko91}}]
	Let $\dom(f)$ be  bounded closed interval $[a, b]$. Let $f: \dom(f) \rightarrow \R$ be a computable function. Then, we say that the time complexity of $f$ on $\dom(f)$ is bounded by a function $t: \dom(f) \times \N \rightarrow \N$ if there exists an oracle TM $M$ which computes $f$ 
	such that for all $x \in \dom(f)$ and all $n>0$, $T_{M}^{\prime}(x,~ n) \leq t(x, ~n)$.
\end{definition}

In other words, as mentioned in \cite{Ko91}, the idea is to measure the time complexity of a function over the reals based on the input and the output precision $2^{-n}$. However, it is usually more convenient to simplify the complexity measure and make it depend only on the output precision. They define the uniform time complexity of $f$ on $G$ with a bound function $t^{\prime}: \N \rightarrow \N$.

However, if we do so, it is critical to realise that if we took $\dom(f)=\R$ in previous definition, for unbounded functions $f$, \cite{Ko91} observes that \textit{``the uniform time complexity does not exist, because the number of moves required to write down the integral part of $f(x)$ grows as $x$ approaches $+\infty$ or $-\infty$''}. 

Therefore, the approach of \cite{Ko91} is to do as follows (the bounds $-2^{X}$ and $2^{X}$ are somewhat arbitrary, but are  chosen here  because the binary expansion of any $x \in\left(-2^{n}, 2^{n}\right)$ has $n$ bits in the integral part):
	
	\begin{definition}[Adapted from \cite{Ko91}]  For functions $f(x)$ whose domain is $\R$, 
		we say that the (non-uniform) time complexity of $f$ is bounded by a function $t^{\prime}: \N^{2} \rightarrow \N$ if the time complexity of $f$ on $\left[-2^{X}, 2^{X}\right]$ is bounded by a function $t: \N^{2} \rightarrow \N$ such that $t(x, ~n) \leq t^{\prime}(X, ~n)$ for all $x \in\left[-2^{X}, 2^{X}\right]$. 
	\end{definition}
	
	We extend the approach to more general functions.  
	
	When $\tu x=(x_{1},\cdots,x_{p})$ and $\tu X= (X_{1},\cdots, X_{p})$, we write
	$\tu x \in [-2^{\tu X}, 2^{\tu X}]$ as short for $x_{1} \in\left[-2^{X_{1}}, 2^{X_{1}}\right]$,  \dots, $x_{p} \in\left[-2^{X_{p}}, 2^{X_{p}}\right]$.
	
	For $n\in\N$, we denote by $\ell(n)$ the number of bits of the binary representation of $n$.

	\begin{definition}[Complexity for real functions: general case, from \cite{Ko91}]   \label{def:bonendroit} Consider a function  $f(x_{1}, \cdots, x_{p}, n_{1}, \cdots, n_{q})$ whose domain is $\R^{p} \times \N^{q}$. 
		We say that the (non-uniform) time complexity of $f$ is bounded by a function $t^{\prime}: \N^{p+q+1} \rightarrow \N$ if the time complexity of $f(\cdot,\cdots,\cdot,\ell(n_{1}),$ $\cdots,\ell(n_{q}))$ on $\left[-2^{X_{1}}, 2^{X_{1}}\right] \times \cdots \times \left[-2^{X_{p}}, 2^{X_{p}}\right] $  
		is bounded by a function: $$t(\cdot,\cdots,\cdot,\ell(n_{1}),\cdots,\ell(n_{q}),\cdot): \N^{p} \times \N \to \N$$ such that 
		$ t(\tu x,\ell(n_{1}),\cdots,\ell(n_{q}), n) \leq t^{\prime}(\tu X,\ell(n_{1}), \cdots,\ell(n_{q}), n)$
	whenever  $\tu x \in \left[-2^{\tu X}, 2^{\tu X}\right].$
	We say that $f$ is polynomial time computable if $t^{\prime}$ can be chosen as a polynomial. 
	We say that a vectorial function is polynomial time computable iff all its components are. 
\end{definition}

The motivation behind those previous definitions is similar to the one in \cite{pouly} and \cite{ICALP2016vantardise}: we want this measure of complexity to extend the usual complexity for functions over the integers, where the complexity of integers is measured with respect to their lengths, and over the reals, where complexity is measured with respect to their approximation.

In particular, in the specific case of a function $f: \N^{d} \to \R^{d'}$, it means there is some polynomial $t': \N^{d+1} \to \N$ such that the time complexity of producing some dyadic approximating $f(\tu m)$ at precision $2^{-n}$ is bounded by $t'\left(\ell(m_{1}),\cdots,\ell(m_{d}),n\right)$, writing $\tu m = (m_{1},\cdots, m_{d})$.

A key observation is the stability of polynomial time under composition:

\begin{lemma}[\cite{Ko91}]\label{lemCompoOK}
The class of polynomial time computable functions is stable under composition.
\end{lemma}

There exist various characterisations of functions computable in polynomial time over the reals. For example, the following is proved in \cite{Ko91} for functions from $[a,b] \to \R$. It could be extended to more general functions.

\begin{theorem}[{Alternative characterisation \cite{Ko91}}]
A function $f:[a,b] \to \R$ is computable in polynomial time  iff there exist polynomial functions $m$ and $q$ and a function $\psi:(\dyadic \cap[a, b]) \times \N \rightarrow \dyadic$ such that
\begin{enumerate}
	\item $m$ is a polynomial modulus function for $f$ on $[a, b]$: $|x-y| \le 2^{-m(n)}$ implies $|f(x)-f(y)| \le 2^{-n}$ for every $x,y \in [a,b]$ and every integer $n$;
	\item for any $d \in \dyadic \cap[a, b]$ and all $n \in \N,|\psi(d, n)-f(d)| \leq 2^{-n}$; 
	\item $\psi(d, n)$ is computable in time $q(\ell(d)+n)$.
\end{enumerate}
\end{theorem}

\section{Length ODEs}\label{lenODE}

\emph{Length ODEs} are a scheme of discrete ODEs, where the derivative is not with steps of 1, as for classical discrete ODEs, but of the successive powers of 2.
 More comprehensive definitions can also be found in \cite{MFCS2019}.

\begin{definition}[Length-ODE \THEPAPIERS]\label{def:linear lengt ODE}
	A function $\tu f : \N \times \R^d \rightarrow \R^{d^{'}}$, $d, d^{'}\in\N$, is length-ODE definable if it corresponds to the
	solution of 
	\begin{equation} \label{SPLode}
		\tu f(0,\tu y) 
		= \tu g(\tu y)   \quad  \text{and} \quad
		\dderivl{\tu f(x,\tu y)} 
		=   \tu u\left(\tu f(x,\tu y), \tu h(x,\tu y),
		x,\tu y\right)
	\end{equation}
\end{definition}

Equation \ref{SPLode} can be formally rewritten as:
	$$ \tu f(x+1,\tu y)= \tu f(x,\tu y) + (\ell(x+1)-\ell(x)) \cdot
\tu u(\tu f(x,\tu y), \tu h(x,\tu y),
x,\tu y)$$
So $\frac{\delta f}{\delta \ell(n)}\left(n,y\right) \neq 0$ when $n=2^m-1$ for some $m\in\N$. The latter is a change of variable on the former: we differentiate with respect to the binary length of an integer $x$ (its number of bits), instead of $x$ itself.
It crucially links the size of a given representation for real numbers and the complexity class its computation is in. 

For polynomial time we consider \emph{Linear} length-ODE, for the following notion of linearity:
\begin{definition}[Essential linearity]
	$\tu g(\tu f(x, ~\tu y), x,
	\tu y)$ is \textit{essentially linear} in $\tu f(x, ~\tu y)$ if
	it is of the form: $
	\tu g(\tu f(x, ~\tu y),~ x,
	~\tu y) =
	\tu A [\tu f(x,~\tu y),~ 
	x,~\tu y] ~ \cdot ~\tu f(x,~\tu y) 
	+ \tu B [\tu f(x,~\tu y), ~
	x,~\tu y ] $
	where $\tu A$ and $\tu B$ depend on some sigmoid on $\tu f(x, \tu y)$.
\end{definition}

Thus, a length ODE is linear when $\tu u$ (in Equation \ref{SPLode}) is \textit{essentially linear} in $\tu f$, 
$\tu g :\R^d\rightarrow\R^{d^{'}}$ and $\tu h: \N \times \R^d \rightarrow \R^{d^{'}}$.

A fundamental fact is that the derivative with respect to the length provides a way to make a change of variables:

\begin{lemma}[{\THEPAPIERS}] \label{myfundobge}
	Assume that \eqref{SPLode} holds. 
	Then $\tu f(x,\tu y)$ is given by 
	$\tu f(x,\tu y)= \tu F(\length{x},\tu y)$
	where $\tu F$ is the solution of the initial value problem 
	\begin{equation} \label{eq:puissance}
		\tu F(1,\tu y)= \tu g(\tu y), \quad and \quad \dderiv{\tu F(t,\tu y)}{t} =\tu u(\tu F(t, \tu y),\tu h(2^{t}-1,\tu y), 2^{t}-1,\tu y).
	\end{equation}
\end{lemma}

This means $\tu f(x,\tu y)$ depends only on the length of its first argument:  $\tu f(x,\tu y)= \tu f(2^{\ell(x)},\tu y)$.
Then \eqref{eq:puissance} can be seen as defining a function (with this latter property) by a recurrence of type
\begin{equation}\label{ourform}
	\tu f(2^{0},\tu y)= \tu g(\tu y), \quad and \quad {\tu f(2^{t+1},\tu y)} =\overline{\tu u}(\tu f(2^{t}, \tu y),\tu h(2^{t}-1,\tu y), 2^{t},\tu y).
\end{equation}
for $\overline{\tu u}$ \textit{essentially linear} in $\tu f(2^{t}, \tu y)$. As recurrence \eqref{eq:puissance} is equivalent to  \eqref{SPLode}:

\begin{corollary}[Linear length ODE presented with powers of $2$]
	A function $\tu f$ is linear length ODE definable iff
	the value of $\tu f(x,\tu y)$ depends only on the length of its first argument
	and satisfies \eqref{ourform}, 
	for some $\tu g$ and $\tu h$, and 
	$\overline{\tu u}$, \textit{essentially linear} in $\tu f(2^{t}, \tu y)$. 
\end{corollary}

	\section{Proof of Lemma \ref{lemNext}, from [4]}\label{AppNext}

		Before giving the proof of the lemma, we admit the following:
	
	\begin{lemma}[\cite{BlancBournezMFCS2023vantardise}] \label{lem:switch:pairs:tanh}
		Let $N\in\N$, $\alpha_{1},  \alpha_{2},\dots, \alpha_{n}\in\N^n$,  pairwise distinct, and, for $1 \le i \le n$ and $1 \le j \leq N$, let $V_{i,j}$ be some constants. Then there exists some  function $$\sendtanh(2^{m},{(\alpha_{i},j) \sendsymbol V_{i,j}})_{i \in \{1,\dots,n\}, j \in \{1,\dots,N\}}$$ in $\manonclasslighttanh$, such as, for all $i \in \{1,\dots,n\}$, $j \in \{1,\dots,N\}$, it maps any $x  \in  [\alpha_{i}-1/4, \alpha_{i}+1/4]$ and any $y \in [j-1/4,j+1/4]$ to a real number at a distance at most $2^{-m}$ of $V_{i,j}$.
	\end{lemma} 
	
	\begin{proof}[Proof of Lemma \ref{lemNext} from \cite{BlancBournezMFCS2023vantardise}]
		We can write  $l = l_0 l^\bullet $ and $r = r_0r^\bullet $, respectively the left part of the tape with respect to the head and the right part of the tape, where $l_{0}$ and $r_{0}$ are the first letters of $l$ and $r$, and $l^\bullet$ and $r^\bullet$  corresponding to the (possibly infinite) word  $l_{-1} l_{-2} \dots$ and $r_{1} r_{2} \dots$ respectively. Notice that $l$ is written in reverse order on the tape.

		\begin{center}
			\begin{tabular}{c c|c|c|c c}
				\hline 
				... & $l^\bullet $ & $l_0 $ & $r_0$ & $ r^\bullet$ & ... \\ 
				\hline 
				\multicolumn{1}{c}{} & 
				\multicolumn{2}{@{}c@{}}{$\underbrace{\hspace*{\dimexpr6\tabcolsep+3\arrayrulewidth}\hphantom{012}}_{l}$} & 
				\multicolumn{2}{@{}c@{}}{$\underbrace{\hspace*{\dimexpr6\tabcolsep+3\arrayrulewidth}\hphantom{3}}_{r}$}
			\end{tabular} 	
		\end{center}

		Consider the transition function working over the representation of the configurations as reals, considering $r_{0} =  \lfloor \base \bar{r}\rfloor$ :
		\begin{align*}
			\mathit{\bar{\Next}}(q,\bar l, \bar r) &= \mathit{\bar{\Next}}(q,\bar{l^\bullet l_0},\bar{r_0r^\bullet}) = (q', \bar{l'}, \bar{r'})\\
			&= \left\{ 
			\begin{array}{ll}
				(q', \bar{l^\bullet l_0 x}, \bar{r^\bullet}) & \mbox{whenever $\delta(q,r_{0}) = (q',x, \rightarrow)$} \\
				(q', \bar{l^\bullet}, \bar{l_0 x r^\bullet}) & \mbox{whenever $\delta(q,r_{0}) = (q',x, \leftarrow)$} \\
			\end{array}
			\right.
		\end{align*}

		\begin{equation}
			\begin{array}{l} \label{textaremplacer}

				\mbox{
					$\bullet$  in the first case ``$\rightarrow$'' : $\bar{l'} = \base^{-1} \bar l + \base^{-1} x $ and $\bar{r'}  = \bar{r^\bullet} = \{\base \bar r\} $} \\
				\mbox{
					$\bullet$ in the second case ``$\leftarrow$'' : $\bar{l'} =\bar{ l^\bullet} = \{\base \bar l \} $ and $\bar{r'} = \base^{-2} \{\base \bar r\}  + \base^{-2} x + \lfloor \base \bar{l}\rfloor / 4$ }\end{array}
		\end{equation}
		
		We introduce the following functions, corresponding respectively to the head moves to the right and to the head moves to the left:
		\begin{align*}
			&\begin{aligned}[t]
				\rightarrow : Q \times \{0,1,3\} &\mapsto \{0,1\} \\
				(q,a) &\mapsto 1 \quad \text{if $\delta(q,a) = (\cdot,\cdot, \rightarrow)$}\\
				(q,a) &\mapsto 0 \quad \text{otherwise}
			\end{aligned}&&
			\begin{aligned}[t]
				\leftarrow:Q \times \{0,1,3\} &\mapsto \{0,1\} \\
				(q,a) &\mapsto 1 \quad \text{if $\delta(q,a) = (\cdot,\cdot, \leftarrow)$}\\
				(q,a) &\mapsto 0 \quad \text{otherwise}
			\end{aligned}\\
		\end{align*}

		\noindent We can rewrite $ \mathit{\bar{\Next}}(q,\bar l, \bar r) = (q',\bar l',\bar r')$ as
		$\bar{l'} = \displaystyle \sum_{q,r_{0}} \left[\rightarrow(q, r_{0}) \left( \frac{\bar{l}}{4} + \frac{x}{4} \right) + \leftarrow(q, r_{0}) \left\{ 4 \bar{l}\right\}\right] $
		and
		$\bar{r'} = \displaystyle \sum_{q,r_{0}} \left[\rightarrow(q, r_{0}) \left\{ 4 \bar{r} \right\} + \leftarrow(q, r_{0}) \left( \frac{\left\{ 4r  \right\} }{4^2} + \frac{x}{4^2} + \frac{\lfloor 4 \bar{l} \rfloor }{4}\right)\right],$
		and, using notation of Lemma \ref{lem:switch:pairs:tanh},
		$q'=\send((q,r) \sendsymbol nextq^{q}_{r})_{q \in Q, r\in \{0,1,3\}}(q,\lfloor 4 \bar{r} \rfloor ).$

	Then, following the intuition of having continuous approximations of discontinuous functions, we can replace $\lfloor \base \bar r \rfloor$ by $\sigma(\base \bar r)$
if we take $\sigma$ as some continuous function that would be affine and values respectively $0$, $1$ and $3$ on $\{0\} \cup [\symboleun,2] \cup [\symboledeux,4]$
(that is to say matches $\lfloor \base \bar r \rfloor$ on this domain). A possible candidate is $\sigma(x)=\sig(1/4,3/4,x) + 2\cdot\sig(9/4,11/4,x)$, with $\sig(a,b,.)$ the sigmoid valuing 0 before $a$ and 1 after $b$.  
Then considering $\xi(x)=x-\sigma(x)$, then  $\xi(\base \bar r)$ would be the same as $\left\{ \base \bar r \right\}$:  
that is, considering  $r_{0}= \sigma(4\bar{r})$,  replacing in the above expression every $\left\{ \base \cdot \right\}$ by $\xi(\cdot)$, and every $\lfloor \cdot \rfloor$ by $\sigma( \cdot)$, and get something that would still work the same, but using only 
continuous functions.

Now, we would like to use analytic functions, and not only continuous functions. It is well-known that an analytic function that equals some affine function on some interval (e.g. on [1,2]) must be affine, and hence cannot be $3$ on $[3,4]$. 
We introduce $\sigtanh(m,a,b,x) \in \manonclasslighttanh$ such that for all integer $m$, $$|\sigtanh(m,a,b,x) - \sig(a,b,x)|\le 2^{-m}.$$
The function $\sigtanh$ allows us to tolerate errors by replacing $\sig(\cdot,\cdot)$ by $\sigtanh(2^{m+c},\cdot,\cdot)$ in the expressions above for $\sigma$ and $\xi$, taking $c$ such that $(3+1/4^{2})3|Q|  \le 2^{c}$. This would just introduce some error at most  $(3+1/4^{2})3|Q| 2^{-c}2^{-m}\le 2^{-m}$.

We can also replace every $\rightarrow(q,r)$ in above expressions for $\bar l'$ and $\bar r'$ by $$\sendtanh(k, (q,r) \sendsymbol \rightarrow(q,r))(q,\sigma(4 \bar r)),$$ for a suitable error bound $k$, and symmetrically for $\leftarrow(q,r)$.
However, if we do so, we still might have some multiplications in the above expressions. 
To avoid that, we give a continous approximation in $\TTtanh \in \manonclasslighttanh$ of the conditional sigmoid function \textit{if$(.,.,.)$}: 
\begin{itemize}
	\item if we take $|d'-0| \le 1/4$, 
	then $|\TTtanh(d',\ell)-0| \le 2^{-m},$ 
	\item and if we take $|d'-1| \le 1/4$, then $|\TTtanh(d',\ell)-\ell| \le 2^{-m}$.
\end{itemize}
		
		We can then also write the above expressions as
		$$
		\begin{array}{ll}
			\bar{l'} =   \sum_{q,r}& \left[ 
			\TTtanh\left(2^{m+c},\sendtanh(2^{2},(q,r) \sendsymbol \rightarrow(q,r)) (q,\sigma(4 \bar r)), 
			\frac{\bar{l}}{4} + \frac{x}{4} \right) \right. \\
			&+ \left.
			\TTtanh\left(2^{m+c},\sendtanh(2^{2},(q,r) \sendsymbol \leftarrow(q,r)) (q,\sigma(4 \bar r)),
			\xi( 4 \bar{l}) \right) \right]
		\end{array}
		$$
		$$
		\begin{array}{ll}
			\bar{r'} =   \sum_{q,r}& \left[ 
			\TTtanh\left(2^{m+c},\sendtanh(2^{2},(q,r) \sendsymbol \rightarrow(q,r)) (q,\sigma(4 \bar r)),
			\xi( 4 \bar{r} ) \right) \right. \\
			&+ \left.
			\TTtanh\left(2^{m+c},\sendtanh(2^{2},(q,r) \sendsymbol \leftarrow(q,r)) (q,\sigma(4 \bar r)),
			\frac{\xi( 4r )}{4^2} + \frac{x}{4^2} + \frac{\sigma(4 \bar{l})}{4} \right) \right]
		\end{array}
		$$
		and still have the same bound on the error.
	\end{proof}

\section{Proof of Theorem \ref{thEmbeddingTM}, from [3, 4]}\label{appPreuveTH}

	\begin{proof} [Proof from \cite{BlancBournezMCU22vantardise, BlancBournezMFCS2023vantardise}]
		We consider a Turing machine $M= \left(Q, \{0,1, 3\}, q_{\init}, \delta, F\right)$. Up to renaming, we assume $Q=\{0,1,\cdots,|Q|-1\}\subset \N$. Let 
		$$ \cdots l_{-k} l_{-k+1} \cdots l_{-1} l_{0} r_0 r_1 \cdots r_n \cdots$$ 
		denote the content of the tape of the Turing machine $M$. In this representation, the head is in front of symbol $r_{0}$, and $l_i, r_{i} \in  \{0,\symboleun,\symboledeux\}$ for all $i$. Furthermore, if the blank symbol $B=0$ appears at some letter $r_{i}$ then it appears at letter $r_{j}$ for all $j \ge i$; if the blank symbol $B=0$ appears at some letter $l_{-i}$ then it appears at letter $l_{-j}$ for all $j \ge i$.
		
			The idea is that such a configuration $C$ can also be encoded by some element $\bar C=\encodageconfiguration(C)=(q, \bar l,\bar r) \in \N \times \R^{2}$, by considering $\bar r= \encodagemot(r)$ and $\bar l=\encodagemot(l)$. 
		
		Denoting the image of $\encodagemot: \Sigma^{\omega} \to \R$ by $\Image$, this even lives in $\Q \times \Image^{2}$. 
		
			\begin{remark}
					Notice that $\Image$ is a Cantor-like set: it corresponds to the rational numbers that can be written using only $1$ and $3$ in base $4$. We write $\Image_{S}$ for those with at most $S$ digits after the point (i.e. of the form $n/4^{S}$ for some integer $n$).
				\end{remark}
		
		We denote by $\delta$ the transition function: given some control state $q \in Q$, and some symbol $r_{0}$, $\delta(q,r_{0})=(q',x,m)$ indicates the new control state $q' \in Q$, the symbol $x$ to be written, and some movement $m \in \{\leftarrow,\rightarrow\}$, corresponding to shifting left or right. For practical reasons, we assume that when $M$ reaches a final state, it stays in the same configuration. 
		
			A configuration $C$ of such TM can be denoted by $C=(q,l,r)$, where 
	\begin{eqnarray*}
		\bar r &=& r_0  r_1 \cdots r_n  \cdots , \\
		\bar l &= & l_{0}  l_{-1}  \cdots  l_{-k}  \cdots 
	\end{eqnarray*} 
	are words over alphabet $\Sigma=\{0,\symboleun,\symboledeux\}$ and where $q \in Q$ denotes  the internal state of $M$. Notice that $r$, the right part, is written from left to right (in the usual way), while $l$, the left part, is written from right to left (in a mirror way). We work with infinite words, but as expected, this covers the case of finite words. Namely, given a finite word $w=w_{1} w_{2}  \dots w_{n}$ on the tape, in control state $q$, with the head in front of symbol $1 \le i \le n$, corresponds to the configuration $C=(q, w_{i-1} w_{i-2} \dots w_{1} 0^{\omega}, w_{i} w_{i+1} \dots w_{n} 0^{\omega})$, that is to say, with eventually the blank symbol $B=0$.

	The idea is that such a configuration $C$ can also be encoded by some element $\encodageconfiguration(C)=(q, \bar l, \bar r) \in \N \times \R^{2}$, by considering 
	\begin{eqnarray*}
		\bar r &=& r_0 \base^{-1} + r_1 \base^{-2} + \cdots +  r_n \base^{-(n+1)} + \cdots , \\
		\bar l &= & l_{0} \base^{-1} + l_{-1} \base^{-2} + \cdots + l_{-k} \base^{-(k+1)} + \cdots 
	\end{eqnarray*}

	In other words,  we encode the configuration of a bi-infinite tape Turing machine $M$ by real numbers using their radix 4 encoding, but using only digits $\{0,\symboleun, \symboledeux\}$. As digit $0$ is used only at the end of the words, this corresponds to some rational numbers for a finite word, while for an infinite word, this corresponds to some real numbers, possibly irrational. 
	If we write: $\encodagemot: \Sigma^{\omega} \to \R$ for the function that maps word $w=w_{0} w_{1} w_{2} \cdots$ to
	$$\encodagemot(w)= w_0 \base^{-1} + w_1 \base^{-2} + \cdots +  w_n \base^{-(n+1)} + \cdots,$$ we can also write: 
	$\encodageconfiguration(C)=\encodageconfiguration(q,l,r)= (q,\encodagemot(l),\encodagemot(r)).$
	
	Here, the transition function $\delta$ leads naturally to a function $\overline{\delta}: C \mapsto C'$ that maps a configuration $C = (q, l, r)$ to a configuration $C' = (q', l', r')$ after one step in $M$. Via the above encoding, this can be seen as a function $\tu f: \N \times \R^{2} \mapsto \N \times \R^{2}$.
		
		Notice that $\tu f$ takes its values in $Q \times [0,1]^{2}$ and it is a piecewise affine map. 
		Actually, if we denote the image of $\encodagemot: \Sigma^{\omega} \to \R$ by $\Image$, it even takes its values in $Q \times \Image^{2}$. 
		
		It remains to prove that $\tu f$ is a PAM:
		
		The function $ {\tu f}$ is basically of the form
		
		$\tu f(q, l, r) = \tu f(q, l_0l^\bullet, r_0r^\bullet) = (q', l', r')$ defined as a definition by case of type:
		\[ (q', l', r') = 
		\left\{ \begin{array}{ll}
			(q', x l_0 l^\bullet , r^\bullet) & \mbox{whenever $\delta(q,r_{0}) = (q', x, \rightarrow)$} \\
			(q', l^\bullet, l_0 x r^\bullet) & \mbox{whenever $\delta(q, r_{0}) = (q',x, \leftarrow)$} \\
		\end{array} \right.
		\]
		
		This rewrites as a function $\bar{\tu f}$ which is similar, working over the representation of the configurations as reals, considering $r_{0} =  \lfloor \base \bar{r}\rfloor$ 
		\begin{align*}
			\mathit{\bar{\tu f}}(q, \bar l, \bar r) &= \mathit{\bar{\tu f}}(q, \bar{ l_0 l^\bullet}, \bar{r_0r^\bullet}) = (q', \bar{l'}, \bar{r'})\\
			&= \left\{ 
			\begin{array}{ll}
				(q', \bar{x l_0 l^\bullet }, \bar{r^\bullet}) & \mbox{whenever $\delta(q, r_{0}) = (q', x, \rightarrow)$} \\
				(q', \bar{l^\bullet}, \bar{l_0 x r^\bullet}) & \mbox{whenever $\delta(q, r_{0}) = (q', x, \leftarrow)$} \\
			\end{array}
			\right.
		\end{align*}

		\begin{equation}
			\begin{array}{l} \label{textaremplacerD}
				\mbox{where $r_{0} =  \lfloor \base \bar{r}\rfloor$ and  \hfill } \\
				\mbox{
					$\bullet$  in the first case ``$\rightarrow$'' : $\bar{l'} = \base^{-1} \bar l + \base^{-1} x $ and $\bar{r'}  = \bar{r^\bullet} = \{\base \bar r\} $} \\
				\mbox{
					$\bullet$ in the second case ``$\leftarrow$'' : $\bar{l'} =\bar{ l^\bullet} = \{\base \bar l \} $ and $\bar{r'} = \base^{-2} \bar{r^\bullet} + \base^{-2} x + \base^{-1} \lfloor \base \bar{l}\rfloor $ }
			\end{array}
		\end{equation}
		
		We introduce the following functions: 
		\begin{align*}
			\rightarrow : Q \times \{0,1,3\} &\mapsto \{0,1\} \\
			(q,a) &\mapsto 1 \quad \text{if $\delta(q,a) = (\cdot,\cdot, \rightarrow)$}\\
			(q,a) &\mapsto 0 \quad \text{otherwise}
		\end{align*}
		corresponding to the head moving to the right, 
		and
		
		\begin{align*}
			\leftarrow:Q \times \{0,1,3\} &\mapsto \{0,1\} \\
			(q,a) &\mapsto 1 \quad \text{if $\delta(q,a) = (\cdot,\cdot, \leftarrow)$}\\
			(q,a) &\mapsto 0 \quad \text{otherwise}
		\end{align*}
		corresponding to the head moving to the left.
		
		Hence, we can rewrite, $ \mathit{\bar{\tu f}}(q, \bar l, \bar r) = (q', \bar l', \bar r')$ as
		$$\bar{l'} = \displaystyle \sum_{q, r_{0}} \left[\rightarrow(q, r_{0}) \left( \frac{\bar{l}}{4} + \frac{x}{4} \right) + \leftarrow(q, r_{0}) \left\{ 4 \bar{l}\right\}\right] $$
		and
		$$\bar{r'} = \displaystyle \sum_{q, r_{0}} \left[\rightarrow(q, r_{0}) \left\{ 4 \bar{r} \right\} + \leftarrow(q, r_{0}) \left( \frac{\left\{ 4 \bar{r}  \right\} }{4^2} + \frac{x}{4^2} + \frac{\lfloor 4 \bar{l} \rfloor}{4} \right)\right].$$
		And $\{\cdot\}$ and $\lfloor\cdot \rfloor$ are PAMs.
	\end{proof}

	The embedding of configurations of TM into reals is $\gammac(C)=\encodageconfiguration(C)$.
	But, without loss of generality, we could have taken $\gammac\left(q,w_{1},w_{2}\right)= \left(q,\gamma\left(w_{1}\right),\gamma\left(w_{2}\right)\right)$ with $\gamma: \Sigma^{*} \to \R$ taken as:

	\begin{itemize}
		\item  the encoding $\gammaentier$  mapping the word $w=a_{1} \dots a_{n}$ to the integer whose binary expression is $w$,
		\item or $\gammacompact$ mapping $w$ to the real number of $[0,1]$ whose binary expansion is $w$,
		\item or  more generally, $\gammacompactk$  or $\gammaentierk$, using base $k$  instead of base $2$ for $k \ge 2$,
		\item or  $\gammacompactpk$ mapping $w$ to $\left(\gammacompactk(w),\ell(w)\right)$. 
	\end{itemize}

	In any case, we can then consider a function $\tu f: X \subseteq \R^{d} \to X$ such that for any configuration $C$, denoting by $C'$ the next configuration, $\tu f\left(\gammac(C)\right)=\gammac(C')$: one step of the TM corresponds to one step in the dynamical system $(X,\tu f)$ with respect to $\gamma$. Then, we have that the diagram of the execution commutes for any number of steps:

	\begin{center}
		\tikzset{every node/.style={align=center}}
		\begin{tikzpicture}[scale=0.9][x=1cm,y=0.8cm]
			\node at (0,0) (a) {$C$};
			\node at (0,-2) (aa) {$\gammac(C)$};
			
			\node at (2,0) (b) {$C'$};
			\node at (2,-2) (bb) {$\gammac(C') $};
			\node at (4,0) (c) {$C''$};
			\node at (4,-2) (cc) {$\gammac(C'') $};
			\node at (6,0) (d) {$C'''$};
			\node at (6,-2) (dd) {$\gammac(C''') $};
			\node at (8,0) (e) {};
			\node at (8,-2) (ee) {};

			\node at (-0.3, -1) (i) {$\gammac$};
			
			\node at (1, 0.3) (ah) {$\vdash$};
			\node at (1, -1.7) (h) {$\tu f$};
			\node at (2.2, -1) (i) {$\gammac$};

			\node at (3, 0.3) (ah) {$\vdash$};
			\node at (3, -1.7) (h) {$\tu f$};
			\node at (4.2, -1) (i) {$\gammac$};
			
			\node at (5, 0.3) (ah) {$\vdash$};
			\node at (5, -1.7) (h) {$\tu f$};
			\node at (6.2, -1) (i) {$\gammac$};
			
			\node at (7, 0.3) (ah) {$\vdash$};
			\node at (7, -1.7) (h) {$\tu f$};

			\draw [->] (a)--(aa);	
			
			\draw [->] (a)--(b);
			\draw [->] (b)--(bb);
			\draw [->] (aa)--(bb);
			
			\draw [->] (b)--(c);
			\draw [->] (c)--(cc);
			\draw [->] (bb)--(cc);
			
			\draw [->] (c)--(d);
			\draw [->] (d)--(dd);
			\draw [->] (cc)--(dd);
			
			\draw [->,dashed] (d)--(e);
			\draw [->,dashed] (dd)--(ee);		
		\end{tikzpicture}
	\end{center}
	
	We can write  $l = l_0 l^\bullet $ and $r = r_0r^\bullet $, respectively the left part of the tape with respect to the head and the right part of the tape, where $l_{0}$ and $r_{0}$ are the first letters of $l$ and $r$, and $l^\bullet$ and $r^\bullet$  corresponding to the (possibly infinite) word  $l_{-1} l_{-2} \dots$ and $r_{1} r_{2} \dots$ respectively. Notice that $l$ is written in reverse order on the tape.
	\begin{figure}[h]
		\begin{center}
			\begin{tabular}{c c|c|c|c c}
				\hline 
				... & $l^\bullet $ & $l_0 $ & $r_0$ & $ r^\bullet$ & ... \\ 
				\hline 
				\multicolumn{1}{c}{} & 
				\multicolumn{2}{@{}c@{}}{$\underbrace{\hspace*{\dimexpr6\tabcolsep+3\arrayrulewidth}\hphantom{012}}_{{l}}$} & 
				\multicolumn{2}{@{}c@{}}{$\underbrace{\hspace*{\dimexpr6\tabcolsep+3\arrayrulewidth}\hphantom{3}}_{r}$}
			\end{tabular} 	
		\end{center}	
	\end{figure}
	The questions related to the existence of trajectories in the dynamical system $(X,\tu f)$ associated with the TM lead to questions about the existence of suitable trajectories in the TM. Specifically, it provides c.e.-hardness of reachability for various classes of dynamical systems.  
Call such a situation a \emph{step-by-step} emulation. However, encodings such as $\gammacompact$, whose image is compact, or on $\encodagemot$ as above,  map intrinsically distinct configurations to points arbitrarily close to each other (a sequence over a compact must have some accumulation point). 
Encodings like $\gammaentier$ do not have a compact image but involve emulations with arbitrarily large integers, which is another issue. 
These observations led to the already mentioned (informal) conjecture that undecidability/hardness may hold only for non-robust systems.
This leads to formally discussing robustness issues for general dynamical systems over  $\R^{d}$, for some $d\in\N$. 
These statements and underlying constructions, which allow the simulation of Turing machines, 
led to solving several open problems: the existence of a universal ODE \cite{ICALP2017}, the proof of the  Turing-completeness of chemical reactions \cite{CMSB17vantardise}, or statements about the hardness of several dynamical systems problems \cite{GracaZhongHandbook}.

	\section{Details, proofs and improvement about the construction of Definition \ref{defBranicky}}\label{detailBra}
	
	\begin{proposition}[\cite{ICALP24}]\label{propBraMarche}
		The solution of the continuous ODE in Definition \ref{defBranicky} emulates continuously the associated discrete ODE. 
	\end{proposition}
	
	\begin{proof}
	The idea is that $\tu y_1$ corresponds to the actual computation of the iterates of $\tu G$ (and hence computes the successive values of $\tu f$)  and $\tu y_2$ acts as a ``memory'' equation. Let us detail how it works.  
	
	\begin{remark}
		We describe here an ``ideal'' computation, as $\theta(x)$ is exactly $0$ when
		$x\leq 0 $, and $r(z)$ is exactly some integer on suitable domains. 
	\end{remark}
	
	Initially, $\tu f\left(0,\tu x\right) = \tu y_{1}\left(0,\tu x\right) = \tu y_{2}\left(0,\tu x\right) = \tu g\left(\tu x\right)$.
	For $t \in [0,1/2]$, we have $$\theta\left(- \sin(2\pi t)\right)=0,$$ and hence 
	$\tu y_2' = 0$, so $\tu y_2$ is fixed and kept at value $\tu g\left(\tu x\right)$ for $t \in \left[0,\frac{1}{2}\right]$. 
	Consequently, for $t \in \left[0,\frac{1}{2}\right]$, $r\left(\tu y_{2}\right)$ is also fixed and kept at value $\tu g\left(\tu x\right)$, and $r(t)$ is also fixed and kept at value $0$. Consequently, on this interval, 
	if we write $C(t) = c\theta\left(\sin(2\pi t)\right)$, then the dynamics of $\tu y_1$ is given by 
	\begin{equation} \label{eq:de:branicky}
		\tu y_1' = C (t) \left(\tu G\left(\tu g(\tu x),0,\tu x\right) - \tu y_1 \right)^3
	\end{equation}
	Then, a key is to analyse such an ODE:
	
	\begin{lemma}[Analysis of ODE \eqref{eq:de:branicky} \cite{ICALP24}] \label{bran:funda}
		The solution $\tu y_{1}(t,\tu x)$ of ODE \eqref{eq:de:branicky} is converging to $G\left(\tu g(\tu x),0,\tu x\right)$ for any initial condition. Furthermore, for any initial condition $\tu y_{1}(0,\tu x) \neq G(\tu g(\tu x),0,\tu x)$, we have $\left\| \tu y_1\left(\frac{1}{2}, \tu x\right) - \tu G\left(\tu g(x),0,\tu x\right) \right\| \leq \frac{\sqrt{2}}{2\sqrt{\int_{0}^{\frac{1}{2}} C(z)dz}}$.
	\end{lemma}
	
	For any $m \in \N$, we can select constant $c$ such that for any initial condition $\tu y_{1}(0,\tu x)$, 
	$$\left\| \tu y_1\left(\frac{1}{2}, \tu x\right) - \tu G\left(\tu g(x),0,\tu x\right) \right\| \le 2^{-m}$$ 
	
	Consequently, $\tu y_1(t,\tu x)$ will approach $\tu G(\tu g(\tu x),0,\tu x)=\tu f(1,\tu x)$ on this interval.
	Thus, $\tu y_1(\frac{1}{2},\tu x) =_\epsilon \tu f(1,\tu x)$ and $\tu y_2(\frac{1}{2},\tu x) = \tu g(\tu x)$, for some $\epsilon > 0$, that we can consider less than $\frac14=2^{-2}$, by selecting a big enough constant $c$ (just taking $m=2$ above). 
	At $t=\frac{1}{2}$, $\tu y_1$ will hence have simulated one step of the discrete ODE. 
	
	Now, for $t \in \left[\frac{1}{2}, 1\right]$ the roles of $\tu y_1$ and $\tu y_2$ are exchanged : $\tu y_1'(t, \tu x) = 0$, so $\tu y_1$ is kept fixed, $\tu y_2$ approaches $r(\tu y_{1})=\tu f(1, \tu x)$, thus 
	$\tu y_1(1,\tu x) =_\epsilon \tu y_2(1, x) =_\epsilon \tu f(1, \tu x)$.
	
	By induction, from the same reasoning, we obtain that, for all $n\in\N$, $\tu y_1(n,\tu x) =_\epsilon \tu y_2(n,\tu x) =_\epsilon \tu f(n,\tu x)$, and actually, we also have, for all $t \in \left[n,n+\frac12\right]$, for any integer $n$:
	$$\tu y_1\left(t+\frac12,\tu x\right) =_\epsilon \tu y_2\left(t,\tu x\right) =_\epsilon \tu f\left(n,\tu x\right).$$ \end{proof}
	
	To implement such continuous ODEs, we must fix a function $\theta(x)$ with the above property. Taking 
	$\relu(x)=\max(0, x)$ would satisfy it, but it is not differentiable, and hence would not lead to an ODE. We could take $\theta(x)=0$ for $x \le 0$, $\exp\left(-\frac{1}{x}\right)$ for $x>0$. The point is that such a function is not real analytic. The base functions considered in $\contClass$  are all real analytic, and real analytic functions are preserved by composition, so we cannot get such a function by composition from our base functions. 
	Thus, we consider the following continuous approximation on $\relu$  (\cite[Lemma 19]{BlancBournezMFCS2023vantardise}):
	
	\begin{lemma}[\cite{BlancBournezMFCS2023vantardise}]\label{lemmeRelu}
		We set $Y(x,2^{m+2}) = \frac{1+\tanh(2^{m+2} x)}{2} $. For all integer $m$, for all $x\in \R$,
		$| \relu(x) - xY(x, 2^{m+2})| \leq 2^{-m}$, where $\relu(x)=\max(0, x)$. 
	\end{lemma}
	
	Having $Y(x,z)=\frac{1+\tanh(4 x z)}{2}$ yields a function in $\contClasslight = \contClasslim \backslash \{\robODE,$ $ \MANONlim\}$ with the same property: we avoid the computation of $2^{m}$ by a substitution of a variable, and using a multiplication. We write $\relutanh(Y,x)$ for $xY(x,z)$: we have $| \relutanh(2^{m},x) -\relu(x)| \leq 2^{-m}$.
	Furthermore, a real analytic function constant on some interval (we assumed it is $0$ for $x\le 0$) is constant. Hence, the above-considered function $\theta(x)$ cannot be real analytic: implementing this trick cannot be done directly using our base functions and only compositions. 
	In Proposition \ref{the:manon:trick}, we do a similar construction, but we improve it to deal with errors and not exact functions $\theta(z)$ and $r(x)$. 
	Furthermore, the purpose of the function $r$ is to correct errors around integers, i.e. around $\Z$: it is possibly around other $\Z\delta = \{ n \delta| n \in \Z\}$ for some $\delta>0$. Then, we assume that, for $z \in \left[n\delta-\frac14\delta,n\delta +\frac12\delta \right]$, $r(z)=n$, for any integer $n \in \Z$, to have a similar reasoning as the one above, where $\delta=1$.

	\begin{proposition}[Simulating a discrete ODE by a continuous ODE \cite{ICALP24}] \label{the:manon:trick}	
		Suppose that, in \eqref{code}, we replace function $\theta(z)$ and function $r(z)$ by some suitable approximations: we take $\theta(x)=\relu(x)$, $\theta_{\epsilon'}(x)$, $r_{\epsilon'(z)}$ such that $\theta(z)=_{\epsilon'}\theta_{\epsilon'}(z)$ and $r(x)=_{\epsilon'}r_{\epsilon'}(x)$,  and take constant $\mathfrak{c}$ big enough. 	Then the solution of the obtained ODE continuously simulates the discrete ODE, with error at most $\epsilon$ if $\epsilon'$ is taken sufficiently small. To guarantee $\epsilon=2^{-n}$, it is sufficient to take $\epsilon'=2^{-q(n)}$  and $\theta_{\epsilon'}(x)=\relutanh(2^{q(n)},x)$ for some polynomial $q$. 
	\end{proposition}

	The main idea is that the constructions always replace every function with another not changing much locally (i.e. changes with an error we control). This is the key that provides a robust ODE as in Definition \ref{defRobODE}, leading to polynomial space complexity. 
	Whenever there is a discrete ODE $\frac{\delta \tu f}{\delta t}(t,\tu x) = \tu u (\tu f(t,\tu x),t,\tu x)$ defining some function $\tu f(t, \tu x)$, one can construct some continuous ODE, using only functions from $\contClasslight$, such that one of its projections provides $\tu f\left(z, t, \tu x\right)$, where $\tu f\left(2^{n}, t, \tu x\right) =_{2^{-n}} \tu f\left(n, \tu x\right)$ whenever $t =_{\frac14} n$, for some $n\in\N$. 
	The simulation of Proposition \ref{the:manon:trick} is computable in polynomial time.
	The sine function is computable in polynomial time (adapting \cite{kawamura2009lipschitz}), so is $\pi$ (consider the Taylor expansion) and the other functions of $\contClasslight$. $\FPtime$ is closed under composition.

\end{document}